\documentclass[a4paper,11pt]{article} 

\usepackage[left=3cm,right=3cm,top=3cm,bottom=3cm]{geometry}
\usepackage{braket}
\usepackage[utf8]{inputenc}  
\usepackage{times}
\usepackage{amsmath, amsthm, amssymb}  
\usepackage{bbm} 
\usepackage{url} 
\usepackage{color}
\usepackage{enumerate}
\usepackage{enumitem}
\usepackage{graphicx} 
\usepackage{mathrsfs} 
\usepackage{stmaryrd}
\usepackage{bm}
\usepackage{amsfonts}
\usepackage{hyperref}
\usepackage{cleveref}
\usepackage{algpseudocode}
\usepackage{mathtools}

\usepackage{algorithm}
\usepackage{algpseudocode}
\usepackage{mathtools}
\usepackage{relsize,exscale}
\usepackage[hang,small]{caption}
\usepackage{float}
\usepackage{amsthm}
\usepackage{multicol}

\DeclareMathOperator*{\argmax}{arg\,max}
\DeclareMathOperator*{\argmin}{arg\,min}

\theoremstyle{plain} 
\newtheorem{lemme}{Lemma}
\newtheorem{corollaire}[lemme]{Corollary}
\theoremstyle{plain} 
\newtheorem{definition}[lemme]{Definition}

\newtheorem{notations}[lemme]{Notation}
\newtheorem{proposition}[lemme]{Proposition}
\newtheorem*{proposition sans num}{Proposition}

\newtheorem{remarque}[lemme]{Remark}
\newtheorem{theorem}[lemme]{Theorem}

\newcommand{\bC}{\mathbbm{C}}
\newcommand{\bF}{\mathbbm{F}}
\newcommand{\bN}{\mathbbm{N}}

\newcommand{\cC}{\mathcal{C}}
\newcommand{\cE}{\mathcal{E}}
\newcommand{\cG}{\mathcal{G}}

\newcommand{\cP}{\mathcal{P}}
\newcommand{\cT}{\mathcal{T}}

\newcommand{\cF}{\mathcal{F}}
\newcommand{\cX}{\mathcal{X}}

\title{Efficient decoding of random errors for quantum expander codes}
\date{}
\author{Omar Fawzi\thanks{ENS de Lyon, UMR 5668 LIP - CNRS - UCBL - INRIA - Universit\'e de Lyon, 69364 Lyon, France. omar.fawzi@ens-lyon.fr} \qquad Antoine Grospellier\thanks{Inria, SECRET Project, 2 Rue Simone Iff, 75012 Paris Cedex, France. antoine.grospellier@inria.fr, anthony.leverrier@inria.fr} \qquad Anthony Leverrier\footnotemark[2]}

\begin{document}

\maketitle

\begin{abstract}
  We show that quantum expander codes, a constant-rate family of quantum LDPC codes, with the quasi-linear time decoding algorithm of Leverrier, Tillich and Z\'emor can correct a constant fraction of random errors with very high probability. This is the first construction of a constant-rate quantum LDPC code with an efficient decoding algorithm that can correct a linear number of random errors with a negligible failure probability. Finding codes with these properties is also motivated by Gottesman's construction of fault tolerant schemes with constant space overhead.
  
  In order to obtain this result, we study a notion of $\alpha$-percolation: for a random subset $E$ of vertices of a given graph, we consider the size of the largest connected $\alpha$-subset of $E$, where $X$ is an $\alpha$-subset of $E$ if $|X \cap E| \geq \alpha |X|$.
\end{abstract}

\section{Introduction}

A major goal of quantum information research is to build a large quantum computer that would exploit the laws of quantum mechanics to solve problems out of reach for classical computers. 
Perhaps the main difficulty of this program lies with the fragility of quantum information. While classical information can easily be stored and manipulated without error, this is not the case of quantum information which will tend to decohere quickly as soon as it interacts with the environment. Finding the right balance between protecting the qubits from decoherence and allowing them to be accessed in order to perform a computation turned out to be an extremely challenging task. 

A natural solution to protect quantum information is to encode the logical qubits into a larger set of physical qubits thanks to quantum error correction techniques.
Classically, good error correcting codes that feature a constant encoding rate, a minimal distance linear in their length and efficient decoding algorithms exist; this is for instance the case of low-density parity-check (LDPC) codes \cite{richardson2008modern}.
The situation for quantum codes is direr. Here, the LDPC condition is particularly appealing for implementation, but finding good quantum LDPC codes has proven difficult. The canonical quantum LDPC code is the toric code due to Kitaev \cite{kitaev2003fault}, which is indeed local since only involving 4-body correlations and can be decoded efficiently thanks to Edmonds algorithm for minimum weight perfect matching \cite{dennis2002topological} (or more recently \cite{DN17}). Unfortunately, the toric code as well as its surface code generalizations display a zero rate, in the sense that the ratio between encoded and physical qubits tends to zero when the code size tends to infinity.
The minimum distance of the toric code, on the other hand, is quite good and grows as $\Theta(n^{1/2})$ for a code of length $n$. This is essentially the best available value for LDPC codes since the current record is $\Omega(n^{1/2} \log^{1/4}n)$ due to Freedman, Meyer and Luo \cite{freedman2002z2} and it is a fascinating theoretical open question to construct quantum LDPC codes with minimum distance $\Omega(n^{1/2+\epsilon})$ for some constant $\epsilon >0$. However, even if for the toric code and its generalizations there exist errors of weight $O(\sqrt{n})$ that cannot be corrected, it turns out that there are efficient decoding algorithms that can, with high probability, correct \emph{random errors} on a linear number of qubits provided the error rate is below some threshold~\cite{dennis2002topological}. 

In this paper, we ask the following natural question: are there quantum LDPC codes with efficient decoding algorithms that display constant rate and can correct a linear number of random errors with very high probability?
This question is also motivated by reducing the space overhead for fault-tolerant quantum computation. One of the crowning theoretical achievements in the field of quantum computation is the \emph{threshold theorem} \cite{aharonov1997fault} which shows that an arbitrarily long computation can be performed with an imperfect implementation (noisy qubits, imperfect gates), provided that the noise level is below some constant threshold. 
This is achieved thanks to \emph{quantum fault-tolerance} techniques, by replacing the target logical circuit using $m$ qubits and containing $T$ gates by a fault-tolerant circuit using $O(m \mathrm{polylog}(mT))$ qubits. In a recent breakthrough paper \cite{gottesman2013fault}, Gottesman has argued that the polylogarithmic space overhead isn't necessary and is in fact a consequence of using zero-rate codes such as the toric code in the fault-tolerant circuit. In contrast, one can achieve \emph{a constant space overhead} by relying on constant-rate quantum LDPC codes satisfying a number of desirable properties. 

The most important such property is that it should be able to efficiently correct with very high probability\footnote{We use the term very high probability to refer to a probability $1 - \epsilon(n)$ where $\epsilon(n)$ is negligible, i.e., $n^c \epsilon(n) \to 0$ as $n \to \infty$ for any constant $c$.} random errors provided the error rate is below some threshold.\footnote{As is standard in the fault tolerance literature, we use the word ``threshold'' to mean that below this value, errors are corrected with high probability. And we do not care about what happens above this value.
} 
A main contender for such a family of codes is obtained from the hypergraph-product construction of Tillich and Z\'emor \cite{tillich2014quantum}. Starting with good classical LDPC codes, the construction yields quantum LDPC codes with constant rate and minimum distance growing like $n^{1/2}$. The reason they cannot directly be plugged in Gottesman's construction is the absence of an efficient decoding algorithm to correct a linear number of random errors. Other code constructions such as hyperbolic codes in 4 dimensions are also a candidate since they combine a constant rate with a polynomial minimum distance \cite{guth2014quantum}, but they lack an efficient decoding algorithm providing the required error suppression, see however \cite{hastings2013decoding,londe2017golden}. Finally, hyperbolic surface codes can also achieve constant rate but only with a logarithmic minimum distance \cite{freedman2002z2,kim2007quantum,zemor2009cayley}. While such codes come with an efficient decoding algorithm, for instance Edmonds' maximum matching algorithm \cite{edmonds1965maximum}, the success probability will not be very high and will scale like $1 - 1/\mathrm{poly}(n)$ because the minimum distance is only logarithmic. As a result, they cannot be used to provide quantum fault tolerance with constant overhead and error thresholds independent of how the computation size scales with the number of qubits.

Our main result is to show that the family of \emph{quantum expander codes} of Ref.~\cite{leverrier2015quantum} combine the following desirable properties: they are constant rate, with a quasi-linear time decoding algorithm that corrects adversarial errors of weight $O(\sqrt{n})$ and corrects random errors of linear size with very high probability provided that the error rate is below some constant value.

\subsection*{Main result and proof techniques}

In this paper, we consider quantum expander codes, a family of constant-rate quantum hypergraph-product codes~\cite{tillich2014quantum}, and show that the efficient decoding algorithm introduced in Ref.~\cite{leverrier2015quantum} can correct a constant fraction of random errors. More precisely, we prove the existence of a threshold for various models of noise, including depolarizing noise, as well as more adversarial models which are relevant for quantum fault-tolerance. 
The decoding algorithm is inspired by the bit-flip decoding algorithm studied by Sipser and Spielman in the context of classical expander codes \cite{sipser1996expander}, and works by finding error patterns of small size that decrease the weight of the syndrome. For this reason, we'll refer to it as ``small-set-flip'' decoding algorithm in the sequel.
Similarly to the classical case, it runs in quasi-linear time.

Our approach for the analysis of the decoding algorithm is inspired by the previous work of Kovalev and Pryadko \cite{kovalev2013fault} who studied the behaviour of the maximum likelihood decoding algorithm (that has exponential running time in general) applied to hypergraph-product codes and established the existence of a threshold for the random independent error model. Here, we use similar techniques to the study of a sub-optimal but efficient decoding algorithm, and also find the existence of a threshold.

Technically, we represent the set of qubits as a graph $\cG=(\{1, \dots, n\},\cE)$ called \emph{adjacency graph} where the vertices correspond to the qubits of the code and two qubits are linked by an edge if there is a stabilizer generator that acts on the two qubits. The approach is then to show that provided the vertices $E$ corresponding to the error do not form large \emph{connected} subsets, the error can be corrected by the decoding algorithm. How large the connected subsets are allowed to be is related to the minimum distance of the code for the maximum-likelihood decoder, or to the maximum size of correctable errors for more general decoders. This naturally leads to studying the size of the largest connected subset of a randomly chosen set of vertices of a graph. This is also called site percolation on finite graphs and its connection with the performance of decoding algorithms is old and was used for example in~\cite{dennis2002topological}, \cite{delfosse2010quantum}, \cite{delfosse2013upper} for the surface codes and in \cite{kovalev2013fault} for the maximum likelihood decoder of hypergraph product codes.

Here, in order to analyse the efficient small-set-flip decoding algorithm for expander graphs, we identify a slightly more complex notion of connectivity as relevant. Namely, instead of studying the size of the largest connected subset of $E$, we study the size of the largest connected $\alpha$-subset of $E$. We say that $X$ is an \emph{$\alpha$-subset} of $E$ if $|X \cap E| \geq \alpha |X|$. Note that for $\alpha = 1$, this is the same as $X$ is a subset of $E$. Our main technical contribution is to show that for some relevant noise models, if the probability of error of each qubit is below some threshold depending on $\alpha$ and the degree of $\cG$, then the probability that a random set $E$ has a connected $\alpha$-subset of size $\omega(\log n)$ vanishes as $n$ grows. As the minimum distance of expander codes is $\Omega(\sqrt{n})$ and the efficient decoding algorithm can correctly decode this many errors, we obtain the claimed result.

In fact, our analysis is not restricted to quantum expander codes and applies to any ``local'' decoding algorithm (see Section \ref{subsec:local} for a definition). We show that a ``local'' decoding algorithm with parameter $\alpha$ can correct errors on a subset of qubits $E$ if the size of the largest connected $\alpha$-subset of $E$ is of size less than $t$, where $t$ is such that the decoding algorithm can correct all errors of weight at most $t$. A ``local'' decoding algorithm here refers to a decoding algorithm where in each step errors on distant qubits are decoded independently, and is satisfied by the small-set-flip decoding algorithm for quantum expander codes. Similarly, the bit-flip decoding algorithm for classical expander codes as well as the maximum likelihood decoders for classical LDPC codes or quantum CSS LDPC codes are local in this sense. 
Our main theorem is the following:

  \begin{theorem}\label{main theorem}
    Consider a quantum expander code with sufficient expansion and the small-set-flip decoding algorithm (\Cref{algo decodage qec0}). Then there exists a probability $p_0 > 0$ and constants $C,C'$ such that if the noise parameter satisfies $p < p_0$, the small-set-flip decoding algorithm corrects a random error with probability at least $1 - C n\left(\frac{p}{p_0}\right)^{C' \sqrt{n}}$.
  \end{theorem}

These results are shown in the independent error model, and also in the local stochastic error model that will be defined in Section~\ref{sec:noise} and is intermediate between independent and adversarial errors, and is particularly relevant in the context of fault tolerance.
The proof of~\Cref{main theorem} can be found at the end of Section~\ref{sec:perc}. We point out that the case of quantum CSS LDPC codes with the (inefficient) maximum likelihood decoder was already studied in Ref.~\cite{kovalev2013fault} for the independent error model and in Ref.~\cite{gottesman2013fault} for local stochastic noise.

The family of quantum expander codes with the small-set-flip decoding algorithm thus satisfies the requirements of the construction of~\cite{gottesman2013fault} for a fault-tolerant scheme with constant space overhead. However, to plug these codes in the construction, we need in addition to analyze errors on the syndrome measurements.  We believe that such syndrome errors can be analyzed similarly as in~\cite{gottesman2013fault}, but a full analysis remains to be done.

The manuscript is structured as follows. In Section \ref{sec:codes}, we recall the constructions of classical and quantum expander codes and describe their respective decoding algorithms. In Section \ref{sec:algos}, we consider a generalized version of the small-set-flip decoding algorithm for quantum expander codes and introduce a notion of ``locality'' for decoding algorithms as well as associated parameters $\alpha$ and $t$. We then prove that the algorithm will correct any error $E$ as long as all connected $\alpha$-subsets are of size at most $t$. In Section \ref{sec:perc}, we study the notion of $\alpha$-percolation with the tools of percolation theory and show that for $t = \omega(\log n)$, there exists a threshold below which a random set $E$ will have connected $\alpha$-subsets of size at most $t$ with high probability. Finally, in Section \ref{sec:appl}, we show that the efficient small-set-flip decoding algorithm for quantum expander codes satisfies the locality condition necessary to apply our results.

%%%%%%%%%%%%%%%%%%%%%%%%%%%%
\section{Classical and Quantum expander codes}
\label{sec:codes}

\subsection{Classical expander codes}
\label{subsec:classical}

A linear classical error correcting code $\cC$ of dimension $k$ and length $n$ is a subspace of $\bF_2^n$ of dimension $k$. Mathematically, it can be defined as the kernel of an $(n-k) \times n$ matrix $H$, called the parity-check matrix of the code: $\cC = \{ x \in \bF_2^n \: : \: Hx = 0\}$. The minimum distance $d_{\min}$ of the code is the minimum Hamming weight of a nonzero codeword: $d_{\min} = \min \{ |x| \: : \: x \in \cC, x\ne 0\}$. Such a linear code is often denoted as $[n,k,d_{\min}]$. 
It is natural to consider families of codes, instead of single instances, and study the dependence between the parameters $n, k$ and $d_{\min}$. In particular, a family of codes has \emph{constant rate} if $k = \Theta(n)$. 
Another property of interest of a linear code is the weight of the rows and columns of the parity-check matrix $H$. If these weights are upper bounded by a constant, then we say that the code is a \emph{low-density parity-check} (LDPC) code~\cite{gallager1962low}. This property is particularly attractive because it allows for efficient decoding algorithms, based on message passing for instance. 

An alternative description of a linear code is via a bipartite graph known as its \emph{factor graph}. Let $G =(V \cup C, \cE)$ be a bipartite graph, with $|V|=n_V$ and $|C|=n_C$. With such a graph, we associate the $n_C \times n_V$ matrix $H$, whose rows are indexed by the vertices of $C$, whose columns are indexed by the vertices of $V$, and such that $H_{cv} = 1$ if $v$ and $c$ are adjacent in $G$ and $H_{cv}=0$ otherwise. The binary linear code $\cC_G$ associated with $G$ is the code with parity-check matrix $H$. The graph $G$ is the \emph{factor graph} of the code $\cC_G$ , $V$ is the set of \emph{bits} and $C$ is the set of \emph{check-nodes}.

It will be convenient to describe codewords and error patterns as subsets of $V$: the binary word $e \in \bF_2^{n_V}$ is described by a subset $E \subseteq V$ whose indicator vector is $e$. Similarly we define the \emph{syndrome} of a binary word either as a binary vector of length $n_C$ or as a subset of $C$:
  \begin{align*}
    \sigma(e) := H e \in \bF_2^{n_C},
    && \sigma(E) := \bigoplus_{v \in E} \Gamma(v) \subseteq C.
  \end{align*}
  In this paper, the operator $\oplus$ is interpreted either as the symmetric difference of sets or as the bit-wise exclusive disjunction depending on whether errors and syndromes are interpreted as sets or as binary vectors.

A family of codes that will be central in this work are those associated to so-called \emph{expander graphs}, that were first considered by Sipser and Spielman in \cite{sipser1996expander}.
\begin{definition}[Expander graph]
  \label{def:exp}
  Let $G = (V \cup C, \cE)$ be a bipartite graph with left and right degrees bounded by $d_V$ and $d_C$ respectively. Let $|V| = n_V$ and $|C| = n_C$. We say that $G$ is $(\gamma_V, \delta_V)$-\emph{left-expanding} for some constants $\gamma_V, \delta_V >0$, if for any subset $S \subseteq V$ with $|S| \leq \gamma_V n_V$, the neighbourhood $\Gamma(S)$ of $S$ in the graph $G$ satisfies $|\Gamma(S)| \geq (1-\delta_V) d_V |S|$. Similarly, we say that $G$ is $(\gamma_C, \delta_C)$-\emph{right-expanding} if for any subset $S \subseteq C$ with $|S| \leq \gamma_C n_C$, we have $|\Gamma(S)| \geq (1-\delta_C) d_C |S|$. Finally, the graph $G$ is said $(\gamma_V, \delta_V, \gamma_C, \delta_C)$-\emph{expanding} if it is both $(\gamma_V, \delta_V)$-left expanding and $(\gamma_C, \delta_C)$-right expanding.  
\end{definition}

Sipser and Spielman introduced \emph{expander codes}, which are the codes associated with (left)-expander graphs. Remarkably these codes come with an efficient decoding algorithm, that can correct \emph{arbitrary} errors of weight $\Omega(n)$ \cite{sipser1996expander}.
\begin{theorem}[Sipser, Spielman \cite{sipser1996expander}]
  \label{thm:SS}
  Let $G = (V \cup C, \cE)$ be a $(\gamma, \delta)$-left expander graph with $\delta < 1/4$. There exists an efficient decoding algorithm for the associated code $\cC_G$ that corrects all error patterns $E \subseteq V$ such that $|E| \leq \gamma (1-2\delta) |V|$.
\end{theorem}
The decoding algorithm called the ``bit-flip'' algorithm is very simple: one simply cycles through the bits and flip them if this operation leads to a reduction of the syndrome weight. Sipser and Spielman showed that provided the expansion is sufficient, such an algorithm will always succeed in identifying the error if its weight is below $\gamma(1-2\delta)|V|$. 
In this paper, however,  we will be interested in the decoding of \emph{quantum} expander codes, that we will review next.

Before that, let us mention for completeness that although finding explicit constructions of highly-expanding graphs is a hard problem, such graphs can nevertheless be found efficiently by probabilistic techniques. Verifying that a given graph is expanding is a hard task, however.
\begin{theorem}[Theorem 8.7 of \cite{richardson2008modern}]
  \label{thm:exist}
  Let $\delta_V, \delta_C$ be positive constants. For integers $d_V > 1/\delta_V$ and $d_C > 1/\delta_C$, a graph $G = (V \cup C, \cE)$ with left-degree bounded by $d_V$ and right-degree bounded by $d_C$ chosen at random according to some distribution is $(\gamma_V, \delta_V, \gamma_C, \delta_C)$-expanding for $\gamma_V, \gamma_C = \Omega(1)$ with high probability.
\end{theorem}

\subsection{Quantum error correcting codes}

A quantum code encoding $k$ logical qubits into $n$ physical qubits is a subspace of $(\bC^2)^{\otimes n}$ of dimension $2^k$. A quantum \emph{stabilizer code} is described by a stabilizer, that is an Abelian group of $n$-qubit Pauli operators (tensor products of single-qubit Pauli operators $X, Y, Z$ and $I$ with an overall phase of $\pm1$ or $\pm i$) that does not contain $- I$. The code is defined as the eigenspace of the stabilizer with eigenvalue $+1$ (\cite{gottesman1997stabilizer}).
A stabilizer code of dimension $k$ can be described by a set of $n-k$ generators of its stabilizer group.

A particularly nice construction of stabilizer codes is given by the CSS construction (\cite{calderbank1996good}, \cite{steane1996error}), where the stabilizer generators are either products of single-qubit $X$ Pauli matrices or products of $Z$ Pauli matrices. The condition that the stabilizer group is Abelian therefore only needs to be enforced between $X$-type generators (corresponding to products of Pauli $X$ operators) and $Z$-type generators. 
More precisely, consider two classical linear codes $\cC_X$ and $\cC_Z$  of length $n$ satisfying $\cC_Z^\perp \subseteq \cC_X$, or equivalently, $\cC_X^\perp \subseteq \cC_Z$. (Here, $\cC_X^\perp$ is the dual code to $\cC_X$ consisting of the words which are orthogonal to all the words in $\cC_X$.) This condition also reads $H_X \cdot H_Z^T=0$, if $H_X$ and $H_Z$ denote the respective parity-check matrices of $\cC_X$ and $\cC_Z$.
The quantum code $CSS(\cC_X, \cC_Z)$ associated with $\cC_X$ (used to correct $X$-type errors and corresponding to $Z$-type stabilizer generators) and $\cC_Z$ (used to correct $Z$-type errors and corresponding to $X$-type stabilizer generators) has length $n$ and is defined as the linear span of $\left\{ \sum_{z \in \cC_Z^\perp} |x +z \rangle \: : \: x \in \cC_X\right\}$, where $\{ |x\rangle \: : \: x\in \bF_2^n\}$ is the canonical basis of $(\bC^2)^{\otimes n}$.
In particular, two states differing by an element of the stabilizer group are equivalent. The dimension of the CSS code is given by $k = \dim (\cC_X/ \cC_Z^\perp) = \dim (\cC_Z/\cC_X^\perp) = \dim \cC_X+ \dim \cC_Z - n$. Its minimum distance is defined in analogy with the classical case as the minimum number of single-qubit Pauli operators needed to map a codeword to an orthogonal one. For the code $CSS(\cC_X, \cC_Z)$, one has $d_{\min} = \min(d_X,d_Z)$ where $d_X = \min \{|E|: E \in \cC_X \setminus {\cC_Z}^{\bot} \}$ and $d_Z = \min \{|E|: E \in \cC_Z \setminus {\cC_X}^{\bot} \}$. We say that  $CSS(\cC_X, \cC_Z)$ is a $[[n,k,d_{\min}]]$ quantum code. In the following, it will be convenient to consider the factor graph $G_X = (V \cup C_X, \cE_X)$ (resp.~$G_Z$) of $\cC_X$ (resp.~of $\cC_Z$). We will denote by  $\Gamma_X$ (resp.~$\Gamma_Z$) the neighbourhood in $G_X$ (resp.~$G_Z$). For instance, if $g \in C_Z$ is an $X$-type generator, that is a product of Pauli $X$ operators, then $\Gamma_Z(g)$ is the set of qubits (indexed by $V$) on which the generator acts non-trivially.

Among stabilizer codes, and CSS codes, the class of quantum LDPC codes stands out for practical reasons: these are the codes for which one can find \emph{sparse} parity-check matrices $H_X$ and $H_Z$. More precisely, such matrices are assumed to have constant row weight and constant column weight. Physically, this means that each generator of the stabilizer acts at most on a constant number of qubits, and that each qubit is acted upon by a constant number of generators. Note, however, that while surface codes exhibit in addition spatial locality in the sense that interactions only involve spatially close qubits (for an explicit layout of the qubits in Euclidean space), we do not require this for general LDPC codes. This means that generators might involve long-range interactions. This seems necessary in order to find constant rate quantum codes with growing minimum distance~\cite{bravyi2009no}.

In this work, we will be concerned with Pauli-type noise, mapping a qubit $\rho$ to $p_{\mathbbm{1}} \rho + p_X X \rho X + p_Y Y \rho Y + p_Z Z \rho Z$, for some $p_{\mathbbm{1}}, p_X, p_Y, p_Z$. Such a noise model is particularly convenient since one can interpret the action of the noise as applying a given Pauli error with some probability. As usual, it is sufficient to deal with both $X$ and $Z$-type errors in order to correct Pauli-type errors, and one can therefore define an error by the locations of the Pauli $X$ and Pauli $Z$ errors.
An \emph{error pattern} is a pair $(e_X, e_Z)$ of $n$-bit strings, which describe the locations of the Pauli $X$ errors, and Pauli $Z$ errors respectively. The syndrome associated with $(e_X, e_Z)$ for the code $CSS(\cC_X, \cC_Z)$ consists of $\sigma_X = \sigma_X(e_X) := H_X e_X$ and $\sigma_Z = \sigma_Z(e_Z) := H_Z e_Z$. 
A decoder is given the pair $(\sigma_X, \sigma_Z)$ of syndromes and should return a pair of errors $(\hat{e}_X, \hat{e}_Z)$ such that $e_X + \hat{e}_X \in \cC_Z^\perp$ and $e_Z + \hat{e}_Z \in \cC_X^\perp$. In that case, the decoder outputs an error equivalent to $(e_X, e_Z)$, and we say that it succeeds.
\\Similarly as in the classical case, it will be convenient to describe $X$-type error patterns and $X$-type syndromes as subsets of the vertices of the factor graph $G_X=(V \cup C_X, \cE_X)$. The error pattern is then described by a subset $E_X \subseteq V$ whose indicator vector is $e_X$ and the syndrome is the subset $\sigma_X(E_X) \subseteq C_X$ defined by $\sigma_X(E_X) := \bigoplus_{v \in E_X} \Gamma_X(v)$. $Z$-type error patterns and $Z$-type syndromes are described in the same fashion using the factor graph $G_Z$.

In the rest of this paper, we will consider special decoding algorithms that try to recover $e_{X}$ and $e_{Z}$ independently. More precisely, a decoding algorithm is given by an $X$-decoding algorithm that takes as input $\sigma_{X}$ and returns $\hat{e}_X$ such that $\sigma_X(\hat{e}_X) = \sigma_X$, and a $Z$-decoding algorithm that takes as input $\sigma_{Z}$ and returns $\hat{e}_Z$ such that $\sigma_Z(\hat{e}_Z) = \sigma_Z$. We note that this special type of decoding algorithm might achieve sub-optimal error probabilities for some error models. In fact, if there are correlations between $X$ and $Z$ errors (for instance in the case of the depolarizing channel where $p_X = p_Y =p_Z$), one can decrease the error probability by trying to recover $e_{X}$ by using both $\sigma_{X}$ and $\sigma_{Z}$. However, for the purpose of this paper, it is sufficient to consider these special decoding algorithms. In addition, as we will consider codes where $H_{X}$ and $H_{Z}$ play symmetric roles, so the $Z$-decoding algorithm is obtained from the $X$-decoding algorithm by exchanging the roles of $X$ and $Z$.

Let us conclude this section by mentioning the \emph{maximum likelihood decoding algorithm} which returns an error $(\hat{e}_X, \hat{e}_Z)$ of minimum Hamming weight with the appropriate syndrome, that is:
\begin{align*}
  \hat{e}_X = \argmin_{\sigma_X(f_X) = \sigma_X} |f_X|, \quad \hat{e}_Z = \argmin_{ \sigma_Z(f_Z) = \sigma_Z} |f_Z|.
\end{align*}
This inefficient algorithm always succeeds provided that the error weights satisfy $|e_X| \leq \left\lfloor (d_X-1)/2\right\rfloor$ and $|e_Z| \leq \left\lfloor (d_Z-1)/2\right\rfloor$.

\subsection{Quantum expander codes}\label{subs:qec}

In this work, we are particularly interested in a family of LDPC CSS codes that features a constant rate and a minimum distance $\Theta(\sqrt{n})$ obtained by applying the hypergraph product construction of Tillich and Z\'emor to classical expander codes. If these expander codes have sufficient expansion, the corresponding quantum code is called \emph{quantum expander code} and comes with an efficient decoding algorithm that corrects arbitrary errors of size linear in the minimum distance. 

The construction is as follows. Let $G = (A \cup B, \cE)$ be a biregular  $(\gamma_A, \delta_A, \gamma_B, \delta_B)$-expanding graph with $\delta_A, \delta_B < 1/6$, and constant left and right degrees denoted $d_A$ and $d_B$. Let us also denote $n_A = |A|$ and $n_B = |B|$ with $n_A < n_B$. According to Theorem \ref{thm:exist}, such graphs can be found efficiently (in a probabilistic fashion) provided that $d_A, d_B \geq 7$.
Let $\cC$ be the classical code associated with $G$, let $d_{\min}(\cC)$ be the minimal distance of $\cC$ and let $H$ be its parity-check matrix (that we assume to be full rank) corresponding to the factor graph $G$. In particular, the weights of rows and columns of $H$ are $d_A$ and $d_B$, respectively. The hypergraph product code of $\cC$ with itself admits the following parity check matrices:
\begin{align*}
  H_X &= \left( I_{n_A} \otimes H, H^T \otimes I_{n_B}\right)\\
  H_Z &= \left( H \otimes I_{n_A}, I_{n_B} \otimes H^T \right).
\end{align*}
It is immediate to check that this defines a legitimate CSS code since 
\begin{align*}
  H_X H_Z^T &=  I_{n_A} \otimes H \cdot ( H \otimes I_{n_A})^T + H^T \otimes I_{n_B} (I_{n_B} \otimes H^T)^T \\
  &= H^T \otimes  H+ H^T \otimes H =0. 
\end{align*}
Moreover, the code is LDPC with generators of weight $d_A + d_B$ and qubits involved in at most $2 \max(d_A, d_B)$ generators.
\\We can describe the factor graphs $G_X$ and $G_Z$ as follows: the set of qubits is indexed by $A^2 \cup B^2$, the set of $Z$-type generators is indexed by $A \times B$ and the set of $X$-type generators is indexed by $B \times A$. The bipartite graph $G_X$ has left vertices $A^2 \cup B^2$, right vertices $A \times B$ and there is an edge between a vertex $(\alpha,a) \in A^2$ (resp. $(b,\beta) \in B^2$) and a vertex $(\alpha,\beta) \in A \times B$ when $a$ (resp. $b$) is in the neighbourhood of $\beta$ (resp. $\alpha$) in $G$. The bipartite graph $G_Z$ has left vertices $A^2 \cup B^2$, right vertices $B \times A$ and there is an edge between a vertex $(\alpha,a) \in A^2$ (resp. $(b,\beta) \in B^2$) and a vertex $(b,a) \in B \times A$ when $\alpha$ (resp. $\beta$) is in the neighbourhood of $b$ (resp. $a$) in $G$.
\\The following theorem summarizes the main properties of this quantum code.
\begin{theorem}[Tillich, Z\'emor \cite{tillich2014quantum}]\label{thm:TZ}
  The CSS code defined above is LDPC with parameters $\left[\left[ n, k, d_{\min}\right]\right]$, where $n = {n_A}^2 + {n_B}^2, k \geq (n_A - n_B)^2$ and $d_{\min} = d_{\min}(\cC)$.
\end{theorem}
Since the graph $G$ is sufficiently expanding, we can apply the results from Ref.~\cite{leverrier2015quantum} and show the existence of an efficient decoding algorithm called ``small-set-flip'' decoding algorithm.
Focusing on $X$-type errors for instance, and assuming that the syndrome $\sigma_X = H_X e_X$ is known, the algorithm cycles through all the $X$-type generators of the stabilizer group (i.e. the rows of $H_Z$), and for each one of them, determines whether there is an error pattern contained in the generator that decreases the syndrome weight.
Assuming that this is the case, the algorithm applies the error pattern (maximizing the ratio between the syndrome weight decrease and the pattern weight). The algorithm then proceeds by examining the next generator. Since the generators have (constant) weight $d_A+d_B$, there are $2^{d_A + d_B}$ possible patterns to examine for each generator. If the graph $G$ has sufficient expansion and if the error weight is small enough, there always exists a generator containing an error pattern decreasing the syndrome weight. It can then be proved that if the error weight is below the value $w_0$ of Eq.~\eqref{eq:w0} below, the decoding algorithm will not stop before reaching a null syndrome, hence corresponding to a codeword. Moreover, because the number of steps of the algorithm is sufficiently low, it is not possible to have reached an incorrect codeword, hence the decoding succeeded. 
\\Let us introduce some additional notations: let $\cX$ be the set of subsets of $V$ corresponding to $X$-type generators: $\cX = \{\Gamma_Z(g): g \in C_Z\} \subseteq \cP(V)$. The indicator vectors of the elements of $\cX$ span the dual code $\cC_Z^\perp$. 
  The condition for successful decoding then asks that there exists a subset $X \subset \cX$ such that
  \begin{align*}
    E \oplus \hat{E} = \bigoplus_{x \in X} x, 
  \end{align*}
  meaning that the remaining error after decoding is trivial, that is equal to a sum of generators.
  At each step, the small-set-flip algorithm tries to flip a subset of some $g \in \cX$. In other words, it tries to flip some element of $\cF := \{F \subseteq x: x \in \cX\}$.

  \begin{algorithm}[H]
    \caption{(Ref.~\cite{leverrier2015quantum}): Small-set-flip decoding algorithm for quantum expander codes
    }\label{algo decodage qec0}
      {\bf INPUT:} $\sigma \subseteq {C}_X$, a syndrome where $\sigma = \sigma_X(E)$ with $E \subseteq {V}$ an error
      \\{\bf OUTPUT:} $\hat{E}\subseteq {V}$, a guess for the error pattern (alternatively, a set of qubits to correct)
      \\{\bf SUCCESS:} if $E \oplus \hat{E} = \bigoplus_{x \in X} x$ for $X \subseteq \cX$, \textit{i.e.} $E$ and $\hat{E}$ are equivalent errors
      
      \hrulefill
      \begin{algorithmic}
        \State{$\hat{E}_0 = 0$ ; $\sigma_0 = \sigma$ ; $i = 0$}
        \While{$\displaystyle \left(\exists F \in \cF:  |\sigma_i|-|\sigma_i \oplus \sigma_X (F)| > 0\right)$}
        \\\State{$\displaystyle F_i = \argmax_{F \in \cF} \frac{|\sigma_i|-|\sigma_i \oplus \sigma_X(F)|}{|F|} \qquad$ // pick an arbitrary one if there are several choices}
        \\\State{$\hat{E}_{i+1} = \hat{E}_i \oplus F_i$}
        \State{$\sigma_{i+1} = \sigma_i \oplus \sigma_X (F_i)$ \qquad // $\sigma_{i+1} = \sigma_X (E \oplus \hat{E}_{i+1})$}
        \State{$i = i+1$}
        \EndWhile
        \\\Return{$\hat{E}_i$}
      \end{algorithmic}
  \end{algorithm}
  
\begin{theorem}[Leverrier, Tillich, Z\'emor \cite{leverrier2015quantum}]\label{thm:LTZ}
  Let $G = (A \cup B, \cE)$ be a $(d_A,d_B)$-biregular  $(\gamma_A, \delta_A, \gamma_B, \delta_B)$-expanding graph with $\delta_A, \delta_B < 1/6$. Letting $d_A$ and $d_B$ be fixed and allowing $n_A, n_B$ to grow, \Cref{algo decodage qec0} runs in time linear in the code length $n=n_A^2 + n_B^2$, and decodes any quantum error pattern of weight less than
  \begin{align}
    \label{eq:w0}
    w_0 = \frac{1}{3(1+d_B)}\min(\gamma_A n_A, \gamma_B n_B).
  \end{align}
\end{theorem}

The analysis above applies to arbitrary errors of weight less than $w_0$. Unfortunately, $w_0 = O(\sqrt{n})$, corresponding to an error rate of $O(1/\sqrt{n})$, which is often insufficient for applications, where a constant error rate is typically required. 
The main goal of this work is to determine the performance of the algorithm above against \emph{random noise} with constant error rate. Such an error model leads to typical errors of size linear in $n$, and Theorem \ref{thm:LTZ} is useless in that regime. The rest of the paper is devoted to show the existence of a threshold for natural models of random noise for \Cref{algo decodage qec0}.

\section{Efficient decoding algorithms for quantum expander codes}
\label{sec:algos}

In this section, we first describe two natural noise models for which we will establish the existence of thresholds for the decoding algorithm of quantum expander codes. 
Then, we will present a variant of the small-set-flip decoding algorithm of Ref.~\cite{leverrier2015quantum}. Finally, we will describe the crucial features of this decoding algorithm, notably its \emph{locality}, and define parameters $\alpha$ and $t$, that will be particularly useful to establish the existence of thresholds.

\subsection{Noise models}
\label{sec:noise}

We think of an error pattern as a pair $(E_X, E_Z)$ of subsets of the set of qubits $V$. A noise model is thus described by a distribution on pairs of subsets of $V$. Letting $(E_X, E_Z)$ be distributed according to such a model, the failure probability of a decoding algorithm that outputs $(\hat{E}_X, \hat{E}_Z)$ is given by 
\begin{align}
&\mathbb{P}((E_X, E_Z) \text{ not equivalent to } (\hat{E}_X, \hat{E}_Z) ) \nonumber \\
&\leq \mathbb{P}(E_X \text{ not equivalent to } \hat{E}_X) + \mathbb{P}(E_Z \text{ not equivalent to } \hat{E}_Z) \ .
\label{eq:decomp_x_z_error}
\end{align}
As the decoders we consider treat $X$ and $Z$ symmetrically, it suffices to focus on $X$-type errors and the analysis of failure probability for identifying $Z$-type errors is exactly symmetrical. For this reason, we will now assume that we deal with $X$-type errors and simply use $E$ for $E_X$.

The most natural error model is that of independent noise, where each qubit is in error independently, with the same probability.
\begin{definition}[Independent noise error model]\label{model erreur}\
  \\Let $V$ be the set of qubits. The error of parameter $p$ is a random variable $E \subseteq V$ such that $\mathbb{P}(E) = p^{|E|} (1-p)^{|V|-|E|}$.
\end{definition}
For example, the $X$-type errors (resp. $Z$-type errors) produced by a depolarizing channel $\rho \mapsto (1 - p) \rho + p/3 (X \rho X + Y \rho Y + Z \rho Z)$ are independent with parameter $2p/3$ since $Y = i XZ = -i ZX$. More generally, a channel $\rho \mapsto p_{\mathbbm{1}} \rho + p_X X \rho X + p_Y Y \rho Y + p_Z Z \rho Z$ satisfies the independent noise condition with parameter $p_X + p_Y$ for the $X$-type errors and parameter $p_Y + p_Z$ for the $Z$-type errors.

The independent noise error model is often too restrictive in practice, notably in the context of fault-tolerance since errors will propagate through the circuit and become correlated. Instead of trying to control such correlations, it is handy to simply allow for any possible error set, but put a bound on the probability of observing an error of a given size. A convenient model was considered by Gottesman in Ref.~\cite{gottesman2013fault} where we ask that such a probability is exponentially small in the error size. 
\begin{definition}[Local stochastic error model]\label{model erreur lc}\
  \\Let $V$ be the set of qubits. The error of parameter $p$ is a random variable $E \subseteq V$ such that for all $F \subseteq V: \mathbb{P}(F \subseteq E) \leq p^{|F|}$. In other words, the location of the errors is arbitrary but the probability of a given error decays exponentially with its weight.
\end{definition}

\subsection{Small-set-flip decoding algorithm for quantum expander codes}
\label{subsec:QAlgo}

Let us now turn to the decoding of quantum codes. As mentioned previously, we focus here on decoding $X$ errors. Given a syndrome $\sigma_X = \sigma_X(E_X)$, the decoding algorithm should output a guess $\hat{E}_X$ which is equivalent to $E_X$, in the sense that $E_X + \hat{E}_X$ belongs to the dual code $\cC_Z^\perp$.

In the description of \Cref{algo decodage qec0}, the idea is to go through all the elements of $C_Z$ (called $X$-type generators and corresponding to products of Pauli $X$ operators), one at a time, and check whether applying any error pattern within this generator would decrease the syndrome weight. If this is the case, then apply the ``best'' such error pattern. Then proceed with the next generator. 
\\It is possible to design \Cref{algo decodage qec}, a slight generalization of \Cref{algo decodage qec0} that depends on an extra-parameter $\beta \in (0,1]$, where a small set $F$ is flipped only if it leads to a decrease of the syndrome weight by at least $\beta d_B |F|$. Note that \Cref{algo decodage qec} is a ``tool'' that we use in order to study \Cref{algo decodage qec0} since any error $E \subseteq V$ corrected by \Cref{algo decodage qec} is corrected by \Cref{algo decodage qec0} (see \Cref{rq:algo}). Requiring a larger value for $\beta$ reduces the number of qubits incorrectly flipped during the decoding procedure. This property will be useful for studying the performance of the algorithm against various random noise models. The drawback, however, is that larger values of $\beta$ mean better expansion (smaller value of $\delta_A$ and $\delta_B$) and therefore larger degrees $d_A$ and $d_B$. The value of $\beta$ should be optimized in order to find the largest possible threshold $p_0$ under which the small-set-flip decoding algorithm will correct random errors, except with negligible probability.

  \begin{algorithm}[H]
    \caption{: Small-set-flip decoding algorithm for quantum expander codes, with parameter $\beta \in (0, 1]$
    }\label{algo decodage qec}
      {\bf INPUT:} $\sigma \subseteq {C}_X$, a syndrome where $\sigma = \sigma_X(E)$ with $E \subseteq {V}$ an error
      \\{\bf OUTPUT:} $\hat{E}\subseteq {V}$, a guess for the error pattern (alternatively, a set of qubits to correct)
      \\{\bf SUCCESS:} if $E \oplus \hat{E} = \bigoplus_{x \in X} x$ for $X \subseteq \cX$, \textit{i.e.} $E$ and $\hat{E}$ are equivalent errors
      
      \hrulefill
      \begin{algorithmic}
        \State{$\hat{E}_0 = 0$ ; $\sigma_0 = \sigma$ ; $i = 0$}
        \While{$\displaystyle \left(\exists F \in \cF:  |\sigma_i|-|\sigma_i \oplus \sigma_X (F)| \geq \beta d_B |F|\right)$}
        \\\State{$\displaystyle F_i = \argmax_{F \in \cF} \frac{|\sigma_i|-|\sigma_i \oplus \sigma_X(F)|}{|F|} \qquad$ // pick an arbitrary one if there are several choices}
        \\\State{$\hat{E}_{i+1} = \hat{E}_i \oplus F_i$}
        \State{$\sigma_{i+1} = \sigma_i \oplus \sigma_X (F_i)$ \qquad // $\sigma_{i+1} = \sigma_X (E \oplus \hat{E}_{i+1})$}
        \State{$i = i+1$}
        \EndWhile
        \\\Return{$\hat{E}_i$}
      \end{algorithmic}
  \end{algorithm}
  \begin{remarque}\label{rq:algo}
    The while loop condition in \Cref{algo decodage qec} is stronger than in \Cref{algo decodage qec0} and this is the only difference between the two algorithms. As a consequence if \Cref{algo decodage qec} corrects an error $E$, then \Cref{algo decodage qec0} corrects $E$.
  \end{remarque}
  Analyzing this modified algorithm can be done similarly as for the original one. For a given quantum expander code, we prove that there exists $\beta_0>0$ such that for $\beta \leq \beta_0$, the small-set-flip decoding algorithm with parameter $\beta$ corrects any adversarial errors of size up to $\Theta(\sqrt{n})$. The value of $\beta_0$ is defined in \Cref{def:correction param}.

  \begin{definition}\label{def:correction param}
    Let $G = (A \cup B, \cE)$ be a $(d_A, d_B)$-biregular $(\gamma_A, \delta_A, \gamma_B, \delta_B)$-left-right-expanding graph. We define $r$ and $\beta_0$ by:
    \begin{align*}
      && r = \frac{d_A}{d_B},
      && \beta_0 = \frac{r}{2 } \left[1 - 4(\delta_A + \delta_B + (\delta_B - \delta_A)^2) \right].
    \end{align*}
  \end{definition}
  Note that $\delta_A, \delta_B < 1/8$ is sufficient to ensure that $\beta_0 > 0$ and there exists such expander graphs as soon as $d_A, d_B \geq 9$ (see \Cref{thm:exist}). As a reminder, the weight of the stabilizer generators is $d_A + d_B$.

  \begin{proposition}\label{thm:beta}
    Let $G = (A \cup B, \cE)$ be a $(d_A, d_B)$-biregular $(\gamma_A, \delta_A, \gamma_B, \delta_B)$-left-right-expanding graph and $r, \beta = \beta_0$ be defined using the associated quantum expander code (\Cref{def:correction param}). The small-set-flip decoding algorithm of parameter $\beta$ described in \Cref{algo decodage qec} can correct any adversarial error of size up to $t_{\mathrm{SSF}(\beta)}$ where:
      \begin{align*}
        t_{\mathrm{SSF}(\beta)} \geq \displaystyle \frac{r \beta}{1+\beta} \min(\gamma_A n_A, \gamma_B n_B).
      \end{align*}
  \end{proposition}

  The proof of \Cref{thm:beta} is provided in \Cref{section proof theo capacite correction algo QEC}.
  Note that \Cref{algo decodage qec} is not deterministic since the $F_i$ are not necessarily unique but if the weight of the error is smaller than $t_{\mathrm{SSF}(\beta)}$, then the error is corrected for any non-deterministic choice of the $F_i$.

  \subsection{Locality of the small-set-flip decoding algorithm}

\label{subsec:local}

  This section is devoted to explicitly state the properties of the small-set-flip decoding algorithm that we need in order to prove \Cref{main theorem}. Recall that our goal is to establish the existence of a threshold: in the random error model of \Cref{model erreur}, if the probability of error is small enough then we correct the error with high probability. An important property of the algorithm that we consider is that it can correct adversarial errors of size $\leq t$ where $t = \Theta(\sqrt{|V|})$ and $V$ is the set of qubits. The second crucial property we will rely on is that the algorithm is local, meaning in words that errors far away in the \emph{adjacency graph} $\cG$ of the code will not interact during the decoding process. The adjacency graph $\cG = (V, \cE)$ associated with the code admits the set of qubits as vertices, and two vertices are adjacent in $\cG$ if the corresponding qubits share an $X$-type or $Z$-type stabilizer generator.

  Let us now define the notion of locality more formally. For an input syndrome $\sigma_{X}(E) \subseteq C_X$ with $E \subseteq V$ an error, the output of \Cref{algo decodage qec} is $\hat{E} = F_0 \oplus \ldots \oplus F_{f-1}$, where $F_i$ is the small set flipped at the $i^{\mathrm{th}}$ round of the algorithm, and we will call $U = E \cup F_0 \cup \ldots \cup F_{f-1}$ the support. We have $\hat{E} \subseteq U$ but $\hat{E} \neq U$ in general.

  \begin{proposition}[Locality of the small-set-flip decoding algorithm]\label{prop:locality}
    Let $E \subseteq V$ be an error and run \Cref{algo decodage qec} on input $\sigma_X(E)$. Consider $K$, a connected component in $\cG$ of $U$ where $U = E \cup F_0 \cup \ldots \cup F_{f-1}$ is the support and $\hat{E} = F_0 \oplus \ldots \oplus F_{f-1}$ is the output.
    \\Then there is a valid execution of \Cref{algo decodage qec} with input $\sigma_X(E \cap K)$, output $\hat{E} \cap K$ and support $K$.
  \end{proposition}

    Let's give the intuition of the proof of \Cref{prop:locality} (the formal proof is given in \Cref{sec:appl}). The execution of \Cref{algo decodage qec} on the input syndrome $\sigma_X(E)$ provides some sets $F_0, \ldots, F_{f-1}$ that we use to define the support $U = E \cup F_0 \cup \ldots \cup F_{f-1}$ and the output $\hat{E} = F_0 \oplus \ldots \oplus F_{f-1}$. We would like to prove that there is a valid execution of the algorithm with input $\sigma_X(E \cap K)$, output $\hat{E} \cap K$ and support $K$ where $K$ is some connected component of $U$ in $\cG$. In other words, we would like to prove that on the input syndrome $\sigma_X(E \cap K)$, \Cref{algo decodage qec} provides some sets $F'_0, \ldots, F'_{f'-1}$ (as well as $\sigma'_0, \ldots, \sigma'_{f'-1}$) which satisfy $K = (E \cap K) \cup F'_0 \cup \ldots \cup F'_{f'-1}$ and $\hat{E} \cap K = F'_0 \oplus \ldots \oplus F'_{f'-1}$.
  \\For simplicity, we first assume that $F_i \subseteq K$ for all $i$ and we define $C_K = \Gamma_X(K) \subseteq C_X$ to be the set of $Z$-type generators which touch at least one qubit of $K$ in the factor graph $G_X$. Let's prove that a valid choice for $f'$, $F'_i$ and $\sigma'_i$ is given by $f' = f$, $F'_i = F_i$ and $\sigma'_i = \sigma_i \cap C_K$. In the factor graph $G_X$, all the neighbours of $C_K$ are in $K \cup \overline{U}$ and all the neighbours of $\overline{C_K}$ are in $\overline{K}$. Hence, the syndrome part included in $C_K$ is the same for an error $W \subseteq V$ and for the error $W \cap K$: $\sigma_X(W \cap K) = \sigma_X(W) \cap C_K$ (\Cref{lem:synd local}). If we apply this formula with $W = E$, we get $\sigma'_0 = \sigma_X(E \cap K)$ which matches the input of \Cref{algo decodage qec} in the conclusion of \Cref{prop:locality}. With $W = F_i$ and our additional hypothesis $F_i \subseteq K$, we get $\sigma'_{i+1} = \sigma'_i \oplus \sigma_X(F_i)$ which matches the manner in which $\sigma_i$ is updated in \Cref{algo decodage qec}. On the other hand, the value of the syndrome $\sigma_i$ determines the algorithm behaviour and, similarly, the value of $\sigma'_i = \sigma_i \cap C_K$ determines the algorithm behaviour on the qubits of $K$ (this is a consequence of \Cref{lemme quantique ok}). As a conclusion, flipping the sets $F_i$ is a valid execution for \Cref{algo decodage qec} on input $\sigma_X(E \cap K)$ and thus \Cref{prop:locality} holds.
  \\In the general case, we have some indices $i$ satisfying $F_i \nsubseteq K$. Since $K$ is a connected component of $U$ in $\cG$, such $i$ always satisfy $F_i \cap K = \varnothing$ and thus flipping $F_i$ does not change the syndrome part included in $C_K$. Hence, with the sets $F'_i$ defined to be the sets $F_i$ such that $F_i \subseteq K$, we can prove \Cref{prop:locality} as previously.

  \bigskip
  
  The concept of locality allows us to treat the errors piece-wise in the sense that if \Cref{algo decodage qec} corrects $E \cap K$ for any connected component $K$ of $U$, then the entire error $E$ is corrected. On the other hand, we know that \Cref{algo decodage qec} corrects any error of size smaller than some $t_{\mathrm{SSF}(\beta)} = \Theta(\sqrt{n})$. We do not have $|E| \leq t_{\mathrm{SSF}(\beta)}$ in general, but using the concept of $\alpha$-subset, we will prove that with high probability $|E \cap K| \leq |K| \leq t_{\mathrm{SSF}(\beta)}$. Hence \Cref{algo decodage qec} corrects $E \cap K$ and, by locality, $E$ is corrected.

  \begin{definition}[$\alpha$-subsets, $\textrm{MaxConn}_{\alpha}(E)$] \label{def:alpha-subset}\ 
    \\An $\alpha$-subset of a set $E \subseteq V$ is a set $X \subseteq V$ such that $|X \cap E| \geq \alpha |X|$. We denote by $\textrm{MaxConn}_{\alpha}(E)$ the maximum size of a connected $\alpha$-subset of $E$.
  \end{definition}
  Note that a $1$-subset is a subset in the usual sense and $\textrm{MaxConn}_{1}(E)$ is the size of the largest connected subset of $E$.
  As a relevant example, a corollary of \Cref{prop:locality} and of Lemma \ref{lemma:alpha subset in algo} is that the set $K$ defined in \Cref{prop:locality} is a connected $\frac{\beta}{1+\beta}$-subset of $E \cap K$.

  \begin{lemme}\label{lemma:alpha subset in algo}
    Let $E \subseteq V$ be an error and run the small-set-flip decoding algorithm of parameter $\beta > 0$ on input $\sigma_X(E)$. Let $\alpha = \frac{\beta}{1+\beta}$ and $U = E \cup F_0 \cup \ldots \cup F_{f-1}$ be the support. Then $U$ is an $\alpha$-subset of $E$.
  \end{lemme}

  \begin{proof}
    For all $i$, the decrease in the syndrome weight $|s_i| - |s_{i+1}|$ satisfies $|s_i| - |s_{i+1}| \geq \beta d_B |F_i|$ and thus $d_B |E| \geq |s_0| \geq |s_0| - |s_f| \geq \beta d_B \sum_i |F_i|$. We conclude that:
    \begin{align*}
      |U|
      = \left|E \cup \left(\bigcup_{i=0}^{f-1} F_i\right)\right|
      \leq |E| + \sum_{i=0}^{f-1} |F_i|
      \leq \frac{1 + \beta}{\beta} |E|
      = \frac{1 + \beta}{\beta} |E \cap U|.
    \end{align*}
  \end{proof}

  All these notions allow us to prove the following general characterization of correctable errors. 
  \begin{proposition}\label{lemme correction}
    Suppose that for some quantum expander code, \Cref{algo decodage qec} with parameter $\beta > 0$ corrects any error $E \subseteq V$ satisfying $|E| \leq t$.
    \\Then any error $E \subseteq V$ satisfying $\textrm{MaxConn}_{\alpha}(E) \leq t$ with $\alpha = \frac{\beta}{1 + \beta}$ is corrected.
  \end{proposition}

  \begin{proof}
    Our goal is to prove that when we apply \Cref{algo decodage qec} on the input $\sigma_X(E)$, the error $E$ and the output $\hat{E} = F_0 \oplus \ldots \oplus F_{f-1}$ are equivalent. Let $U = E \cup F_0 \cup \ldots \cup F_{f-1}$ be the support, then we can decompose $E$ as $E = \biguplus_K (E \cap K)$ and $\hat{E}$ as $\hat{E} = \biguplus_K (\hat{E} \cap K)$, where the index $K$ of these disjoint unions goes through the set of connected components of $U$. In order to prove that $E$ and $\hat{E}$ are equivalent, we will prove the sufficient condition that $E \cap K$ and $\hat{E} \cap K$ are equivalent for any $K$. By \Cref{prop:locality}, there is a valid execution of \Cref{algo decodage qec} with input $\sigma_X(E \cap K)$, output $\hat{E} \cap K$ and support $K$. By \Cref{lemma:alpha subset in algo}, $K$ is an $\alpha$-subset of $E \cap K$ and thus $K$ is a connected $\alpha$-subset of $E$. If we suppose $\textrm{MaxConn}_{\alpha}(E) \leq t$ then $|E \cap K| \leq |K| \leq t$ and we can conclude that $E \cap K$ and $\hat{E} \cap K$ are equivalent because \Cref{algo decodage qec} corrects $E \cap K$.
  \end{proof}

  \begin{remarque}
    A similar analysis has been done in \cite{kovalev2013fault} and \cite{gottesman2013fault} for the maximum likelihood decoding algorithm of quantum LDPC codes and we could also apply the same arguments for the bit-flip algorithm of classical expander codes. However, this is less relevant since this algorithm can already correct adversarial errors of linear size. In fact, the maximum likelihood decoding algorithm and the bit-flip algorithm are both local in the sense of \Cref{prop:locality}.
  \end{remarque}

  %%%%%%%%%%%%%%%%%%%%%%

  \section{Properties of random errors: $\alpha$-percolation}
  \label{sec:perc}

  In order to establish Theorem \ref{main theorem} at the end of this section, we first study a specific model of site percolation that is relevant for the analysis of our decoding algorithms. More specifically, for a graph $\cG = (V, \cE)$ with degree upper bounded by some constant $d$, and a probabilistic model for choosing a random subset $E \subseteq V$, site percolation asks about the size of the largest connected subset of vertices of $E$. Here, instead of looking for a connected subset of $E$, we instead ask about the size $\textrm{MaxConn}_{\alpha}(E)$ of the largest $\alpha$-subset of $E$ which is connected. 

  The simplest model for the choice of $E$ is the i.i.d.~model where each site (or vertex) is occupied independently with probability $p$. In this model, a standard result in percolation theory is that for the infinite $d$-regular tree, if $p > \frac{1}{d-1}$ there is an infinite connected subset of $E$ and if $p < \frac{1}{d-1}$ all connected subsets are finite~\cite{lyons1992random}. In the setting of finite graphs, for a random $d$-regular graph the maximum connected subset of $E$ is of size $O(\log n)$ for $p < \frac{1}{d-1}$ and of size $\Omega(n)$ for $p > \frac{1}{d-1}$~\cite{janson2009percolation}. For general graphs with degree bounded by $d$, \cite{kovalev2013fault}~showed that for $p < \frac{1}{d-1}$ the size of the largest connected component of $E$ in $\cG$ has size $O(\log |V|)$. We will also consider the more general \emph{local stochastic model} where the probability of a subset $E$ is upper bounded by $p^{|E|}$.

  \begin{theorem}[$\alpha$-percolation] 
    \label{thm:percolation}
    Let $\cG = (V, \cE)$ a graph with degree upper bounded by $d$. Let $\alpha \in (0,1]$ and let $t\geq 1$ be an integer. Let
      \begin{align}
        \label{eq:pls}
        p_{\mathrm{ls}} = \left(\frac{2^{-h(\alpha)}}{(d-1)(1+\frac{1}{d-2})^{d-2}}\right)^{\frac{1}{\alpha}} \ ,
      \end{align}
      where $h(\alpha) = - \alpha \log_2 \alpha - (1-\alpha) \log_2(1-\alpha)$ is the binary entropy function.
      Then, for an error random variable $E$ satisfying the local stochastic noise property (Definition~\ref{model erreur lc}) with parameter $p < p_{\mathrm{ls}}$, we have
      \begin{align}
        \label{eq:probls}
        \mathbb{P}[\textrm{MaxConn}_\alpha(E) \geq t] \leq  C |V|\left(\frac{p}{p_{\mathrm{ls}}}\right)^{\alpha t}
      \end{align}
      where $1/C = (1-2^{h(\alpha)/\alpha}p) \left(1-\left(\frac{p}{p_{\mathrm{ls}}} \right)^{\alpha} \right)$.  
      In the special case, where vertices are chosen independently with probability $p < \frac{\alpha}{d-1}$, the result can be slightly improved to 
      \begin{align}
        \label{eq:probiid}
        \mathbb{P}[\textrm{MaxConn}_\alpha(E) \geq t] \leq |V| \left(\frac{d-1}{d-2}\right)^{2} \frac{q^t}{1-q} \ ,
      \end{align}
      where $q = (1-p)^{d-1-\alpha} p^{\alpha} 2^{h(\alpha)}(d-1)(1+\frac{1}{d-2})^{d-2} $.
  \end{theorem}

  Note that~\eqref{eq:probiid} shows that if $p$ is chosen small enough so that $q < 1$ the size of the largest component will be $O(\log |V|)$ with high probability. It is simple to see that choosing $p < p_{\mathrm{ls}}$ leads to $q < 1$ but one can even take $p$ slightly larger. For example, in the case $\alpha = 1$, we have $q < 1$ as soon as $p < \frac{1}{d-1}$ even though $p_{\mathrm{ls}} = \frac{1}{(d-1)(1+\frac{1}{d-2})^{d-2}}$. More generally, $q$ as a function of $p$ is increasing in the interval $p \in [0, \frac{\alpha}{d-1}]$ and thus there is a value $p_{\mathrm{iid}} \in [p_{\mathrm{ls}}, \frac{\alpha}{d-1}]$ such that $q  = 1$ (see~\Cref{lem:threshold-iid}) and any $p < p_{\mathrm{iid}}$ leads to a $q < 1$.
  \begin{proof}
    We focus here on the local stochastic noise bound. The finer analysis of the independent noise model in~\eqref{eq:probiid} is left to Appendix~\ref{section percolation}. 

    Let $\cC_s(\cG)$ be the set of connect sets of vertices of size $s$ in $\cG$.
    Applying a union bound, we obtain
    \begin{align*}
      \mathbb{P}[\textrm{MaxConn}_\alpha(E) \geq t]
      & = \mathbb{P}\left[\exists s \geq t, \exists X \in \cC_s(\cG): |X \cap E| \geq \alpha  |X|\right]\\
      & \leq \sum_{s \geq t} \sum_{X \in \cC_s(\cG)} \mathbb{P}\left[|X \cap E| \geq \alpha  |X|\right].
    \end{align*}
    Let us first consider the quantity $\mathbb{P}\left[|X \cap E| \geq \alpha  |X| \right]$:
    \begin{align*}
      \mathbb{P}\left[|X \cap E| \geq \alpha  |X| \right]
      & \leq \sum_{m \geq \alpha s}\  \sum_{\substack{Y \subseteq X:\\|Y| = m}} \mathbb{P}\left[X \cap E = Y \right]\\
      & \leq \sum_{m \geq \alpha s}\  \sum_{\substack{Y \subseteq X:\\|Y| = m}} \mathbb{P}\left[Y \subseteq E\right]\\
      & \leq \sum_{m \geq \alpha s}\  \sum_{\substack{Y \subseteq X:\\|Y| = m}} p^m \quad \text{from the definition of the noise model}\\
      & \leq \sum_{m \geq \alpha s}\  \binom{s}{m} p^m \ .
    \end{align*}
    Using Lemma~\ref{lem:boundbinomial}, we get
    \begin{align*}
      \mathbb{P}\left[|X \cap E| \geq \alpha  |X| \right] 
      & \leq \sum_{m \geq \alpha s}\  \left(2^{\frac{s}{m} h\big( \frac{m}{s}\big) } p\right)^m \\
      & \leq \sum_{m \geq \alpha s}\  \left(2^{ h(\alpha)/\alpha } p\right)^m  \quad \text{since} \, (x \mapsto h(x)/x) \, \text{is non increasing on} \, [\alpha,1] \\
      & \leq \frac{(2^{ h(\alpha)/\alpha } p)^{\alpha s}}{1-2^{ h(\alpha)/\alpha } p}. 
    \end{align*}
    We now need an upper bound on the number of connected sets of size $s$ in a graph with degree upper bounded by $d$. For this we use Corollary~\ref{lemme nb cc approx} saying that $|\cC_s(\cG)| \leq |V| K(d)^s$ with $K(d) = (d-2) \left(1+\frac{1}{d-2}\right)^{d-2}$.
    Then, we get 
    \begin{align*}
      \mathbb{P}[\textrm{MaxConn}_\alpha(E) \geq t]
      & \leq \sum_{s \geq t} \sum_{X \in \cC_s(\cG)} \mathbb{P}\left[|X \cap E| \geq \alpha  |X|\right)] \\
      & \leq \frac{1}{1-2^{ h(\alpha)/\alpha } p} \sum_{s \geq t} \left|\cC_s(\cG)\right| 2^{ h(\alpha) s} p^{\alpha s} \quad \text{by the union bound}\\
      & \leq \frac{|V|}{1-2^{ h(\alpha)/\alpha } p} \sum_{s \geq t}\left(K(d) 2^{h(\alpha)} p^\alpha\right)^{s} \\ 
      & \leq  \frac{|V|}{1-2^{ h(\alpha)/\alpha } p} \frac{(K(d) 2^{h(\alpha)} p^\alpha)^t}{1- K(d)2^{h(\alpha)} p^\alpha} \ .
    \end{align*}
    Observing that $p_{\mathrm{ls}} = \left(\frac{1}{K(d) 2^{h(\alpha)}}\right)^{1/\alpha}$, we obtain the desired result.
  \end{proof}

  \noindent Using \Cref{thm:beta,lemme correction} and \Cref{thm:percolation}, we are now ready to prove \Cref{main theorem}:

  \begin{proof}[Proof of \Cref{main theorem}]
      Using \Cref{rq:algo}, it is sufficient to prove \Cref{main theorem} for \Cref{algo decodage qec} instead of \Cref{algo decodage qec0}. Under the condition that the expansion of the initial graph $G$ is good enough, the parameter $\beta = \beta_0$ from \Cref{def:correction param} is positive.
    We let $p_0 = p_{\mathrm{ls}}$ as defined in~\eqref{eq:pls} where $\cG$ is the adjacency graph of the quantum expander code and $d = d_B^2 + 2 d_B (d_A - 1)$ is an upper bound on the degree of $\cG$.
    Then for $p < p_0$, for an error $E$ chosen according to a local stochastic noise, and fixing $t = t_{\mathrm{SSF}(\beta)} = \Omega(\sqrt{n})$ as defined in \Cref{thm:beta}, we obtain
    \begin{align*}
      \mathbb{P}[\textrm{MaxConn}_{\alpha}(E) \geq t] \leq C n \left(\frac{p}{p_0}\right)^{\alpha t} \ .
    \end{align*}
    As a result, applying \Cref{lemme correction}, $E$ is corrected except with probability at most $C n \left(\frac{p}{p_0}\right)^{\alpha t}$ and this gives us the desired results.
  \end{proof}

  \begin{remarque}\label{rq optimal threshold}
    It is an interesting question of percolation theory to find the maximum value of $p$ for which the largest connected $\alpha$-subset is of size $O(\log n)$. This will depend on the graph $\cG$ but as mentioned earlier, it is known that for random $d$-regular graphs and $\alpha = 1$, $p = 1/(d-1)$ is the threshold for the appearance of large connected sets (see \cite{janson2009percolation}, Theorem 3.5). In \Cref{prop:ub-percolation}, we show that for $p > \frac{1}{(d-1)^{\frac{1}{\alpha}}}$, the largest connected $\alpha$-subset is of size $\Omega(n^{1-\epsilon})$ and thus Eq.~\eqref{eq:pls} can be improved by a factor of at most $(\frac{2^{-h(\alpha)}}{e})^{\frac{1}{\alpha}}$.
  \end{remarque}

  %%%%%%%%%%%%%%%%%%%%%%%
  \section{The small-set-flip decoding algorithm is local (proof of \Cref{prop:locality})}
  \label{sec:appl}

  In this section we show that the small-set-flip decoding algorithm (\Cref{algo decodage qec}) is local in the sense of \Cref{prop:locality} (the intuition of the proof has been given after the statement of the proposition).
  \begin{notations}
    Throughout this section we keep the following notations:
    \begin{itemize}
    \item The sets $\cF, E, K, U, F_0, \ldots, F_{f-1}$ and $\hat{E}$ are defined as in \Cref{algo decodage qec,prop:locality}.
    \item $C_K = \Gamma_X(K) \subseteq C_X$ is the set of $Z$-type generators which touch at least one qubit of $K$ in the factor graph $G_X$.
    \item For $\sigma \subseteq C_X$ and $F \in \cF$: $\displaystyle \Delta(\sigma,F) = \frac{|\sigma|-|\sigma \oplus \sigma_{X}(F)|}{|F|}$ is the function optimized by \Cref{algo decodage qec}.
    \end{itemize}
  \end{notations}
  
  \begin{lemme}\label{lem:synd local}
    Let $W \subseteq U$ (later we will take either $W = E$ or $W = F_i$) then $\sigma_X(W \cap K) = \sigma_X(W) \cap C_K$.
  \end{lemme}
  \begin{proof}
    We decompose $\sigma_X(W) \cap C_K$ using the partition $W = (W \cap K) \uplus (W \cap \overline{K})$ and the distributivity:
    \begin{align*}
      \sigma_X(W) \cap C_K
      & = \left[\sigma_X(W \cap K) \oplus \sigma_X(W \cap \overline{K})\right] \cap C_K\\
      & = (\sigma_X(W \cap K) \cap C_K) \oplus (\sigma_X(W \cap \overline{K}) \cap C_K).
    \end{align*}
    The first term of this sum satisfies $\sigma_X(W \cap K) \cap C_K = \sigma_X(W \cap K)$ because $\sigma_X(W \cap K) \subseteq \Gamma_X(K) = C_K$. On the other hand, since $K$ is a connected component of $U$ in $\cG$, it holds that $\Gamma_X(U \setminus K) \cap \Gamma_X(K) = \Gamma_X(U \setminus K) \cap C_K = \varnothing$ and thus the second term of the sum satisfies $\sigma_X(W \cap \overline{K}) \cap C_K = \varnothing$ because $W \cap \overline{K} \subseteq U \setminus K$.
\end{proof}
  
  \begin{lemme}\label{lemme quantique ok}
    Let $F \in \cF$ then:
    \begin{enumerate}[label=(\roman*)]
    \item \label{eq:optim1} $\Delta(\sigma \cap C_K,F) \leq \Delta(\sigma,F)$,
    \item \label{eq:optim2} If $F \subseteq K$ then $\Delta(\sigma \cap C_K,F) = \Delta(\sigma,F)$.
    \end{enumerate}
  \end{lemme}

  \begin{proof}
    Rewrite $\Delta(\sigma,F)$ in a more friendly way using the definition of $\oplus$:
    \begin{align}\label{eq:synd decreasing}
      \Delta(\sigma,F) = \frac{2 |\sigma_X(F) \cap \sigma| - |\sigma_X(F)|}{|F|}.
    \end{align}
    The proof of \cref{eq:optim1} is straightforward using \cref{eq:synd decreasing}: let $F \in \cF$, we have:
      \begin{align*}
        \Delta(\sigma \cap C_K,F)
        & = \frac{2 |\sigma_X(F) \cap \sigma \cap C_K| - |\sigma_X(F)|}{|F|} \qquad \text{ by \cref{eq:synd decreasing}}\\
        & \leq \frac{2 |\sigma_X(F) \cap \sigma| - |\sigma_X(F)|}{|F|}\\
        & = \Delta(\sigma,F) \qquad \text{ by \cref{eq:synd decreasing}}
      \end{align*}
      Under the assumption $F \subseteq K$ of \cref{eq:optim2}, we have $\sigma_X(F) \cap C_K = \sigma_X(F)$. As a consequence, the inequality in the proof of \cref{eq:optim1} is an equality and thus \cref{eq:optim2} holds.
  \end{proof}

  We are now ready to prove \Cref{prop:locality}.

  \begin{proof}[Proof of Prop.~\ref{prop:locality}]
    We would like to prove that there is a valid execution of \Cref{algo decodage qec} with input $\sigma_X(E \cap K)$, output $\hat{E} \cap K$ and support $K$. This means that there exists some sets $F'_0, \ldots, F'_{f'-1} \subseteq V$ such that if we define $\sigma'_0 = \sigma_X(E \cap K)$ and $\sigma'_{j+1} = \sigma'_j \oplus \sigma'_X(F'_j)$ for $j \in \llbracket0;f'-1\rrbracket$ then we have the following properties:
    \begin{enumerate}[label=(\roman*)]
    \item \label{it:valid exec1} The set $\hat{E} \cap K$ satisfies $\hat{E} \cap K = F'_0 \oplus \ldots \oplus F'_{f'-1}$.
    \item \label{it:valid exec2} The set $K$ satisfies $K = (E \cap K) \cup F'_0 \cup \ldots \cup F'_{f'-1}$.
    \item \label{it:valid exec3} For $j \in \llbracket0;f'-1\rrbracket$, $F'_j$ maximises $F \mapsto \Delta(\sigma'_j,F)$ over $\cF$.
    \item \label{it:valid exec4} For $j \in \llbracket0;f'-1\rrbracket$, $\Delta(\sigma'_j,F'_j) \geq \beta d_B$.
    \item \label{it:valid exec5} For all $F \in \cF$, $\Delta(\sigma'_{f'},F) < \beta d_B$.
    \end{enumerate}

    In the adjacency graph $\cG$, two qubits in the support of the same generator are connected. As a consequence, since the set $F_i$ is included in the support of a generator and since $K$ is connected, either $F_i \subseteq K$ or $F_i \cap K = \varnothing$.
    \\Consider $i_0 < \ldots < i_{f'-1} \in \llbracket 0;f-1 \rrbracket$ the steps of the algorithm such that $F_{i_k} \subseteq K$ and let's prove that the sets $F'_0 = F_{i_0}, \ldots, F'_{f'} = F_{i_{f'-1}}$ are appropriate.
    By definition we have:
    \begin{align}
      \label{eq:disj Fi1} & F_i \cap K = F_i & \text{if $i \in \{i_0, \ldots, i_{f'-1}\}$}\\
      \label{eq:disj Fi2} & F_i \cap K = \varnothing & \text{if $i \notin \{i_0, \ldots, i_{f'-1}\}$}
    \end{align}
    Using \cref{eq:disj Fi1,eq:disj Fi2}, we can easily prove \cref{it:valid exec1,it:valid exec2}:
    \begin{enumerate}[label=(\roman*)]
    \item [\ref{it:valid exec1}]
      $\displaystyle
      \hat{E} \cap K
      = \left(\bigoplus_{i = 0}^{f-1} F_i\right) \cap K
      = \bigoplus_{i = 0}^{f-1} (F_i \cap K)
      = \bigoplus_{j = 0}^{f'-1} F'_j.
      $
    \item [\ref{it:valid exec2}] $\displaystyle \bigcup_{j = 0}^{f'-1} F'_j = \bigcup_{i = 0}^{f-1} (F_i \cap K)$ and using the distributivity:
      \\$\displaystyle K = U \cap K = (E \cap K) \cup F'_0 \cup \ldots \cup F'_{f'-1}$.
    \end{enumerate}
    By hypothesis, we further have a valid execution of \Cref{algo decodage qec} with input $\sigma_X(E)$, output $\hat{E}$ and support $U$. This implies:
    \begin{enumerate}[label=(\roman*')]\setcounter{enumi}{2}
    \item \label{it:valid exec3'} For $i \in \llbracket0;f-1\rrbracket$, $F_i$ maximises $F \mapsto \Delta(\sigma_i,F)$ over $\cF$.
    \item \label{it:valid exec4'} For $i \in \llbracket0;f-1\rrbracket$, $\Delta(\sigma_i,F_i) \geq \beta d_B$.
    \item \label{it:valid exec5'} For all $F \in \cF$, $\Delta(\sigma_{f},F) < \beta d_B$.
    \end{enumerate}
    We also need the property of \cref{eq:synd equal} where we have defined $i_{f'} = f$:
    \begin{align}\label{eq:synd equal}
      & \sigma'_{j} = \sigma_{i_{j}} \cap D_K \quad j \in \llbracket 0; f'\rrbracket.
    \end{align}
    Note that \cref{it:valid exec3} (resp. \cref{it:valid exec4}, resp. \cref{it:valid exec5}) is a direct consequence of \Cref{lemme quantique ok}, \cref{eq:synd equal} and \cref{it:valid exec3'} (resp. \cref{it:valid exec4'}, resp. \cref{it:valid exec5'}).
    \\To complete the proof, it remains to prove \cref{eq:synd equal}. By \Cref{lem:synd local}:
    \begin{align*}
      & \sigma_0 \cap D_K = \sigma_X(E) \cap D_K = \sigma_X(E \cap K),\\
      & \sigma_{i+1} \cap D_K = \left(\sigma_i \oplus \sigma_X(F_i)\right) \cap D_K = (\sigma_i \cap D_K) \oplus \sigma_X(F_i \cap K)
      \qquad i \in \llbracket0;f-1\rrbracket.
    \end{align*}
    If $i \in \{i_0, \ldots, i_{f'-1}\}$ then by \cref{eq:disj Fi1}: $\sigma_{i+1} \cap D_K = (\sigma_i \cap D_K) \oplus \sigma_X(F_i)$.
    \\If $i \notin \{i_0, \ldots, i_{f'-1}\}$ then by \cref{eq:disj Fi2}: $\sigma_{i+1} \cap D_K = \sigma_i \cap D_K$.
    \\In summary:
    $$\left\{
    \begin{aligned}
      & \sigma'_0 = \sigma_X(E \cap K)\\
      & \sigma'_{j+1} = \sigma'_j \oplus \sigma_X(F_{i_j}) \quad j \in \llbracket 0; f'-1\rrbracket
    \end{aligned}
    \right.
    \quad
    \left\{
    \begin{aligned}
      & \sigma_0 \cap D_K = \sigma_X(E \cap K)\\
      & \sigma_{i+1} \cap D_K = (\sigma_i \cap D_K) \oplus \sigma_X(F_i) \quad & i \in \{i_0, \ldots, i_{f'-1}\}\\
      & \sigma_{i+1} \cap D_K = \sigma_i \cap D_K \quad & i \notin \{i_0, \ldots, i_{f'-1}\}      
    \end{aligned}
    \right.$$
    \Cref{eq:synd equal} is a direct corollary of the following property that we can prove by induction:
    \begin{align*}
      & \sigma'_{0} = \sigma_{0} \cap D_K = \ldots = \sigma_{i_{0}} \cap D_K,\\
      & \sigma'_{j+1} = \sigma_{i_j + 1} \cap D_K = \ldots = \sigma_{i_{j+1}} \cap D_K \quad j \in \llbracket 0; f'-1\rrbracket.
    \end{align*}
  \end{proof}

  \appendix

  %%%%%%%%%%%%%%%%%%%%%%%

  \section{Numerical bounds}

  In this section, we look at the numerical values of the bounds on the threshold that we derived in the main text. We obtain a threshold $p_{\mathrm{ls}} = 2.70.10^{-16}$ for the small-set-flip decoding algorithm of quantum expander codes with errors in the locally stochastic noise model.
  \\By \Cref{thm:exist}, there exists a family of $(\gamma_A, \delta_A, \gamma_B, \delta_B)$-expanders for $\gamma_A, \gamma_B = \Omega(1)$ with left degree $d_A = 38$, right degree $d_B = 39$, $\delta_A = 1/d_A + \epsilon$ and $\delta_B = 1/d_B + \epsilon$ for $\epsilon$ arbitrarily small. From the construction of \Cref{subs:qec} we get a family of quantum expander codes. The degree of the associated adjacency graph is $d = d_B^2 + 2 d_B (d_A - 1) = 4407$ and by \Cref{def:correction param,thm:beta,lemme correction} and \Cref{thm:percolation}, the small-set-flip algorithm with parameter $\beta = \frac{d_A}{2 d_B} \left[1 - 4(\delta_A + \delta_B + (\delta_B - \delta_A)^2) \right] \approx 0.386$ corrects random errors in the locally stochastic noise model of probability $p < p_{\mathrm{ls}}$ with high probability where
  \begin{align*}
    && \alpha = \frac{\beta}{1 + \beta} \approx 0.278,
    && p_{\mathrm{ls}} = \left(\frac{2^{-h(\alpha)}}{(d-1)(1+\frac{1}{d-2})^{d-2}}\right)^{\frac{1}{\alpha}} \approx 2.70 \cdot 10^{-16}.
  \end{align*}
  Note that the threshold $p_{\mathrm{iid}}$ in the independent error model is essentially the same: $p_{\mathrm{iid}} - p_{\mathrm{ls}} \approx 1 \cdot 10^{-27}$.
  \\From a practical point of view, this theoretical value of $2.70 \cdot 10^{-16}$ is too small to be relevant. We believe, however, that in practice the threshold will be much better. Take as an example the toric code: our analysis leads to a threshold of $8.1 \cdot 10^{-4}$ in the independent error model with ideal syndrome extraction, whereas it is known that the actual threshold is around $10 \%$ (see e.g.~Ref.~\cite{wang2009threshold}). Notice furthermore that the graph expansion arguments are generally pessimistic since they are based on averaging (see for example the discussion of Ref.~\cite{richardson2008modern} example 8.14, about the correction capacity of the bit-flip algorithm for classical expander codes). In addition the threshold for $\alpha$-percolation given in \Cref{thm:percolation} is valid for any graph $\cG$ of degree bounded by $d$. In particular, the case $\alpha = 1$ shows that the threshold is the smallest when the graph $\cG$ is randomly sampled \cite{janson2009percolation}. The point is that the number of connected sets in that kind of graphs is close to the upper bound of \Cref{lemme nb cc}. In contrast, the adjacency graphs $\cG$ we considered are not random and seem to contain far less connected sets than the upper bound that we relied on. As a consequence, if one were able to improve the bound of \Cref{lemme nb cc}, it would automatically lead to a better threshold as can be seen in the proof of \Cref{thm:percolation}. One last argument is that one can improve the analysis of the ``propagation'' of errors: suppose that during the bit flip algorithm, a bit $v$ with no error is flipped, then it would mean that at least half of the neighbours of $v$ in the graph $\cG$ are in error. This shows that not all the connected sets allow the error to propagate, and thus the bound of \Cref{lemme nb cc} would be improved by only counting this kind of connected sets. Note also the similarity with the setup of ``bootstrap percolation'' (see Ref.~\cite{adler1991bootstrap}). We believe that some arguments from this field can be used to improve our threshold.

  %%%%%%%%%%%%%%%%%%%%%%%

  \section{Proofs}

  %%%%%%%%%%%%%%%%%%%%%%%

  \subsection{Proof of \Cref{thm:beta}}\label{section proof theo capacite correction algo QEC}

  We follow the proof of \cite{leverrier2015quantum}. An element $g \in C_Z$ is called a generator. An error $E \subseteq V$ is called reduced if $|E|$ is minimal among the errors of $E + {\cC_Z}^{\bot}$.
  Following Ref.~\cite{leverrier2015quantum}, we introduce a notion of \emph{critical} generator which corresponds to a generator $g$ with $\Gamma_Z(g)$ containing a set of qubits that can be flipped to decrease the syndrome weight. More precisely, a critical generator as defined below can always be used to decrease the syndrome weight, but there might exist noncritical generators that also contain a subset of qubits that lead to a decrease of the weight of the syndrome, when flipped. 

  \begin{definition}[Definition 6 of \cite{leverrier2015quantum}]
    \label{def:crit}
    A generator $g \in C_Z$ is said to be \emph{critical} for $E \subseteq V$ if:
    \begin{align*}
      \Gamma_Z(g) = x_a \uplus \overline{x}_a \uplus \chi_a \uplus x_b \uplus \overline{x}_b \uplus \chi_b
    \end{align*}
    where
    \begin{itemize}
    \item $\Gamma_Z(g) \cap A^2 = x_a \uplus \overline{x}_a \uplus \chi_a$ and $\Gamma_Z(g) \cap B^2 = x_b \uplus \overline{x}_b \uplus \chi_b$;
    \item $x_a, x_b \subseteq E$ and $\overline{x}_a, \overline{x}_b \subseteq V \setminus E$;
    \item for all $v_a \in x_a, v_b \in x_b, \overline{v}_a \in \overline{x}_a$ and $\overline{v}_b \in \overline{x}_b$:
      \begin{itemize}
      \item $E \cap \Gamma_X[\Gamma_X(v_a) \cap \Gamma_X(v_b)] = \{v_a, v_b\}$
      \item $E \cap \Gamma_X[\Gamma_X(\overline{v}_a) \cap \Gamma_X(\overline{v}_b)] = \varnothing$
      \item $E \cap \Gamma_X[\Gamma_X(v_a) \cap \Gamma_X(\overline{v}_b)] = \{v_a\}$
      \item $E \cap \Gamma_X[\Gamma_X(\overline{v}_a) \cap \Gamma_X(v_b)] = \{v_b\}$
      \end{itemize}
    \item $x_a \cup x_b \neq \varnothing, |\chi_a| \leq 2 \delta_B d_B$ and $|\chi_b| \leq 2 \delta_A d_A$.
    \end{itemize}
  \end{definition}

  \begin{lemme}[Lemma 7 of \cite{leverrier2015quantum}]\label{existance of a critical generator}
    For an error $E \subseteq V$ such that $0 < |E| \leq \min(\gamma_A n_A, \gamma_B n_B)$, there exists a critical generator for $E$.
  \end{lemme}

  \begin{lemme}[Modified version of Lemma 8 of \cite{leverrier2015quantum}]\label{existance of a critical generator and error}
    Let $r$ and $\beta = \beta_0$ be defined as in Definition \ref{def:correction param}. 
    Let $E \subseteq V$ be a reduced error such that $0 < |E| \leq r \min(\gamma_A n_A, \gamma_B n_B)$ then there exists an error $F \in \cF$ with
    \begin{align*}
      |\sigma_{X}(E)|-|\sigma_{X}(E \oplus F)| \geq \beta d_B |F|.
    \end{align*}
  \end{lemme}

  \begin{proof}

     For convenience, we use a notion of reduced cardinality. Given an error $E \subseteq V = A^2 \uplus B^2$, $\|E\|$ is defined by:
  \begin{align*}
    \|E\| = \frac{|E \cap A^2|}{d_B} + \frac{|E \cap B^2|}{d_A}.
  \end{align*}
  It is straightforward to check that for $E \subseteq V$:
  \begin{align}\label{eq:norm prop}
    d_A \|E\| \leq |E| \leq d_B \|E\|.
  \end{align}
  Let $E_R$ be the error of $E + {\cC_Z}^{\bot}$ which minimizes $\|E_R\|$. The size of $E_R$ satisfies:
  \begin{align*}
    0 <
    |E_R|
    \leq d_B \|E_R\|
    \leq d_B \|E\|
    \leq \frac{d_B}{d_A} |E|
    \leq \min(\gamma_A n_A, \gamma_B n_B).
  \end{align*}
  We exploit \Cref{existance of a critical generator} applied to $E_R$ and prove that if we flip $F = x_A \uplus x_B$ in the generator promised by the lemma then $|\sigma_{X}(E)|-|\sigma_{X}(E \oplus F)| \geq \beta d_B |F|$.

  We use the notations of the proof of Lemma 8 of \cite{leverrier2015quantum}:
    \begin{align*}
      x = \|x_a\| = |x_a|/d_B && z = \|\chi_a\| = |\chi_a|/d_B && \overline{x} = 1-x-z = \|\overline{x}_a\| = |\overline{x}_a|/d_B\\
      y = \|x_b\| = |x_b|/d_A && t = \|\chi_b\| = |\chi_b|/d_A && \overline{y} = 1-y-t = \|\overline{x}_b\| = |\overline{x}_b|/d_A
    \end{align*}
    For $F = x_A \uplus x_B$, we have $\|F\| = x + y$ and:
    \begin{align*}
      |\sigma_{X}(E)|-|\sigma_{X}(E \oplus F)|
      & \geq \left[ |x_a| |\overline{x}_b| + |\overline{x}_a| |x_b| - |x_a| |\chi_b| - |x_b| |\chi_a|\right]\\
      & = d_A d_B \left[ x \overline{y} + \overline{x} y - x t - y z\right]\\
      & = d_A d_B \left[ x (1-y-t) + (1-x-z) y - x t - y z\right]\\
      & = d_A d_B \left[ x + y - 2xy - 2 x t - 2 y z\right]\\
      & \geq d_A d_B \left[ x + y - 2xy - 4 x \delta_A - 4 y \delta_B\right]\\
      & = d_A d_B \left[ x(1-4\delta_A) + y(1-4\delta_B) - 2xy\right]
    \end{align*}
    Moreover $0 < x + y$ because $x_A \uplus x_B \neq \varnothing$ and $x + y \leq 1$ because $E_R$ minimizes $\|E_R\|$ over $E + {\cC_Z}^{\bot}$.
    \\Let $\displaystyle f(x,y) = \frac{1}{x + y} \left[ x(1-4\delta_A) + y(1-4\delta_B) - 2xy\right]$ defined for $x,y$ such that:
    \begin{align*}
      0 \leq x \leq 1, && 0 \leq y \leq 1, && 0 < x + y  \leq 1.
    \end{align*}
    The function analysis below shows that $f(x,y) \geq \beta/r$ and thus \Cref{existance of a critical generator} holds using \cref{eq:norm prop}:
    \begin{align*}
      |\sigma_{X}(E)|-|\sigma_{X}(E \oplus F)|
      \geq d_A d_B f(x,y) \|F\|
      \geq \beta d_B |F|.
    \end{align*}
    The remaining point is to show that $f(x,y) \geq \beta/r = \frac{1}{2} \left[1 - 4(\delta_A + \delta_B + (\delta_B - \delta_A)^2) \right]$:
    \\The partial derivative of $f$ with respect to $x$ is equal to:
    \begin{align*}
      \frac{\partial f}{\partial x}(x,y) = \frac{y}{(x+y)^2} \left[ 4(\delta_B - \delta_A) - 2y\right].
    \end{align*}
    \begin{itemize}
    \item If $y \leq 2(\delta_B - \delta_A)$ then $x \mapsto f(x,y)$ is non-decreasing and thus:
      \begin{align*}
        f(x,y) \geq f(0,y) = 1 - 4 \delta_B \geq \beta/r.
      \end{align*}
    \item Otherwise, $x \mapsto f(x,y)$ is non-increasing and thus:
      \begin{align*}
        f(x,y) \geq f(1-y,y) = 2y^2 + y(4 \delta_A - 4 \delta_B - 2) + 1 - 4 \delta_A.
      \end{align*}
      $f(x,y)$ is then lower bounded by $\beta/r$ the minimum of $y \mapsto f(1-y,y)$ on $[0,1]$ reached for $y = 1/2 + \delta_B - \delta_A$.
    \end{itemize}
  \end{proof}

  \begin{proof}[Proof of \Cref{thm:beta}]
    Let $E_0 \subseteq V$ with $|E_0| \leq \displaystyle \frac{r \beta}{1+\beta} \min(\gamma_A n_A, \gamma_B n_B)$ and run \Cref{algo decodage qec} with parameter $\beta$ on $E_0$. Set $E_i = E_0 \oplus \hat{E}_{i} = E_0 \oplus \left(\bigoplus_{k = 0}^{i-1} F_i\right)$ for $i \in \llbracket 1; f \rrbracket$, where $f$ is the number of steps of the algorithm (the algorithm terminates because the syndrome weight decreases strictly at each step and must remain non negative). Since $|\sigma_X(E_i)|-|\sigma_{X}(E_{i+1})| \geq \beta d_B |F_i|$, we have:
    \begin{align*}
      \beta d_B \sum_{i=0}^{f-1} |F_i|
      \leq \sum_{i=0}^{f-1} |\sigma_X(E_i)|-|\sigma_{X}(E_{i+1})|
      = |\sigma_X(E_0)|-|\sigma_{X}(E_{f})|
      \leq d_B |E_0|,
    \end{align*}
    and:
    \begin{align*}
      |E_f| = \left|E_0 \oplus \left(\bigoplus_{i = 0}^{f-1} F_i\right)\right| \leq |E_0| + \sum_{i = 0}^{f-1} |F_i| \leq \frac{1 + \beta}{\beta} |E_0| \leq r \min(\gamma_A n_A, \gamma_B n_B).
    \end{align*}
    Let $E_R$ be the reduced error of $E_f + {\cC_Z}^{\bot}$ then $|E_R| \leq |E_f| \leq r \min(\gamma_A n_A, \gamma_B n_B)$ and \Cref{algo decodage qec} stops on $s_f = \sigma_X(E_f) = \sigma_X(E_R)$. We conclude by \Cref{existance of a critical generator and error} that $|E_R| = 0$ and therefore the error is corrected.
  \end{proof}

  %%%%%%%%%%%%%%%%%%%%%%%%%
  \newpage

  \subsection{Results for percolation}\label{section percolation}

  We start by stating useful bounds on binomial coefficients.
  \begin{lemme}
    \label{lem:boundbinomial}
    For $k, n$ integers, we have:
    \begin{align*}
      \binom{n}{k} \leq 2^{n h(k/n)},
    \end{align*}
    where $h(x) = - x \log_2(x) - (1-x) \log_2(1-x)$ is the binary entropy function.
  \end{lemme}

  \begin{lemme}[Chernoff bound for the binomial distribution]
    \label{chernoff}
    Let $p \in [0;1], s \in \bN$ and $k \geq s p$, then we have:
    \begin{align*}
      \sum_{m \geq k} \binom{s}{m} p^m (1-p)^{s-m} \leq 2^{-s D\left(\frac{k}{s} \| p\right)}.
    \end{align*}
  \end{lemme}

  \Cref{lemme nb cc} gives an upper bound on the number of connected sets of size $s$ in a graph with an upper bound on the degree. Note that for a rooted tree of degree bounded by $d$, each vertex which is not the root has at most $d-1$ sons. Using this remark and results from \cite{pah2015combinatorial}, Lemma \ref{lemme nb cc} can be proved following the proof of \cite{uehara1999number}.

  \begin{lemme}[\cite{pah2015combinatorial},\cite{uehara1999number}]\label{lemme nb cc}
    Let $\cG = (V,\cE)$ be a graph of degree bounded by $d \geq 2$. For $s \geq 1$, the number of connected sets of size $s$ satisfies $\displaystyle |\cC_s(\cG)| \leq |V| \frac{d}{s[s(d-2)+2]} \binom{s(d-1)}{s-1}$.
  \end{lemme}

  \begin{proof}
    Let $\cT(s)$ be the set of labelled rooted trees with $s$ vertices such that the maximum outdegree of the root is $d$ and the maximum outdegree of the other nodes is $d-1$. An element of $\cT(s)$ is a directed graph where the root has at most $d$ sons and the other vertices have at most $d-1$ sons. The directed edges whose head is some node $v$ are injectively labeled with labels in $\llbracket 1;d \rrbracket$ if $v$ is the root and with labels in $\llbracket 1;d-1 \rrbracket$ if $v$ is not the root. Note that in the underlying undirected graph, a node which is not the root can have $d$ neighbours. Using Ref.~\cite{pah2015combinatorial}, we have $|\cT(s)| = R(d-1,s-1,d) := \frac{d}{s(d-2)+2} \binom{s(d-1)}{s-1}$ where $R(a,b,c)$ are the Raney numbers.
    \\From the graph $\cG$ we construct the oriented graph $\cG_0$ where each non-oriented edge of $\cG$ has been replaced by two opposite oriented edges. Similarly than for the trees, we fix some labelling of $\cG_0$: the directed edges whose head is some node $v$ are injectively labeled with labels in $\llbracket 1;\deg^+(v) \rrbracket$ where $\deg^+(v)$ is the outdegree of $v$.
    \\Now let $v \in V$ and let $\cC_s(v)$ be the set of connected sets $X \in \cC_s(\cG)$ with $v \in X$. For a given $X \in \cC_s(v)$, there is at least one spanning tree $T_0$ of $X$. We then get a labelled rooted tree $T \in \cT(s)$ by fixing $v$ as the root and fixing the labelling of $T$ using the labelling of $T_0$ in $\cG_0$. Note that the labelling for $T$ is not exactly the same labelling than $T_0$ since the labels in $T$ are in $\llbracket 1;d-1 \rrbracket$ (for edges whose head is not the root) whereas the label $d$ is allowed for $T_0$. The key point is that when we consider a vertex $v_1 \neq v$ in $T_0$ then this vertex has a father $v_0$. Let $i$ be the label in $T_0$ of the oriented edge from $v_1$ to $v_0$ and $j$ be the label in $T_0$ of another edge whose head is $v_1$. The label  in $T$ of this edge is then $j$ if $j < i$ and $j-1$ if $j > i$. Conversely given $T \in \cT(s)$, $T$ can be obtained by this procedure using at most one connected set $X \in \cC_s(v)$. Thus $|\cC_s(v)| \leq |\cT(s)|$.
    \\When we consider the sum $\sum_{v \in V} |\cC_s(v)|$, each connected set of $\cC_s(\cG)$ is counted $s$ times. Therefore,
    \begin{align*}
      |\cC_s(\cG)| = \frac{1}{s} \sum_{v \in V} |\cC_s(v)| \leq |V| \frac{d}{s[s(d-2)+2]} \binom{s(d-1)}{s-1}.
    \end{align*}
    
  \end{proof}

  \begin{corollaire}\label{lemme nb cc approx} 
    Let $\cG = (V,\cE)$ be a graph of degree bounded by $d \geq 2$. For $s \geq 1$, the number of connected sets of size $s$ satisfies $|\cC_s(\cG)| \leq |V| K(d)^s$ where $K(d) = (d-1)\left(1+\frac{1}{d-2}\right)^{d-2}$.
  \end{corollaire}

  \begin{proof}
    From \Cref{lemme nb cc}, we have:
    \begin{align*}
      |\cC_s(\cG)|
      & \leq |V| \frac{d}{s[s(d-2)+2]} \binom{s(d-1)}{s-1}\\
      & = |V| \frac{d}{[s(d-2)+1][s(d-2)+2]} \binom{s(d-1)}{s}\\
      & \leq |V| \binom{s (d-1)}{s} \qquad \text{since } d \geq 2 \text{ and } s \geq 1\\
      & \leq |V| 2^{s(d-1)h(1/(d-1))} \\
      &= |V| \left((d-1) \left(1+\frac{1}{d-2}\right)^{d-2} \right)^{s}.
    \end{align*}
  \end{proof}

  We now prove the improved bound~\eqref{eq:probiid} on the probability of $E$ to be a connected $\alpha$-subset of size $\geq t$. The calculation is similar to that in~\cite[Theorem 3]{kovalev2013fault} who considered $\frac{1}{2}$-percolation in their proof.
  \begin{proof}[Proof of Eq.~\eqref{eq:probiid}]

    \bigskip
    \noindent For $v \in V, X \subseteq V, E \subseteq V$ and $s \in \bN$, let us define:
    \begin{align*}
      & \cC_s(v) = \{X \in \cC_s(\cG): v \in X\},\\
      & \partial X = \Gamma(X) \setminus X, \text{ the boundary of } X,\\
      & A(E,\alpha,s,v) = \{X \in \cC_s(v): |X \cap E| \geq \alpha |X|, \partial X \cap E = \varnothing\}.
    \end{align*}
    In other words, the sets $A(E,\alpha,s,v)$ are the connected sets $E$ containing $v$ such that $|X \cap E| \geq \alpha |X|$ (and whose boundary does not intersect $E$).
    The reason for the second condition in the definition of $A(E, \alpha, s,v)$ is to capture the idea that the sets $X \in A$ are of maximum size: indeed, increasing the size of a set $X \in A(E, \alpha, s,v)$ would only decrease the fraction of vertices in $X \cap E$. 
    
    By induction on the size $|X|$ of a connected set $X \subseteq V$, it is straightforward to show that $|\partial X| \leq (d - 2) |X| + 2$.

    By definition, we have:
    \begin{align*}
      \mathbb{P}[\textrm{MaxConn}_{\alpha}(E) \geq t]
      & = \mathbb{P}\left[\exists s \geq t, \exists X \in \cC_s(\cG): |X \cap E| \geq \alpha  |X|\right]\\
      & = \mathbb{P}\left[\exists s \geq t, \exists X \in \cC_s(\cG): |X \cap E| \geq \alpha  |X|, \partial X \cap E = \varnothing \right]\\
      & = \mathbb{P}\left[\exists s \geq t, \exists v \in V: A(E,\alpha,s,v) \neq \varnothing\right]\\
      & \leq \sum_{s \geq t} \sum_{v \in V} \mathbb{P}\left[ A(E,\alpha,s,v) \neq \varnothing\right] \ .
    \end{align*}
    Note that we used the fact that if there exists $X \in \cC_s(\cG)$ with $s \geq t$ such that $|X \cap E| \geq \alpha |X|$, then we can construct $X'$ by repeatedly adding vertices in $\partial X \cap E$ until we cannot do so. Thus, $X'$ satisfied $\partial X \cap E = \varnothing$. In addition, we have $|X' \cap E| - |X \cap E| = |X'| - |X|$ is the number of vertices added in this process, which implies that $|X' \cap E| \geq \alpha |X'|$.

    For $v \in V$ and $s \in \bN$, we have:
    \begin{align*}
      & \mathbb{P}\left[ A(E,\alpha,s,v) \neq \varnothing\right] \leq \sum_{X \in \cC_s({v})} \sum_{m \geq \alpha s} \binom{s}{m} p^m (1-p)^{s-m+|\partial X|}.
    \end{align*}
    Note furthermore that for any $p$, we have
    \begin{align}
      \label{eqn:alpha=1}
      \mathbb{P}\left[ A(E,1,s,v) \neq \varnothing\right] = \sum_{X \in \cC_s({v})} p^s (1-p)^{|\partial X|} \leq 1.
    \end{align}

    %%%%%%%%

    For $p < \frac{1}{d-1}$, we have:
    \begin{align*}
      &\mathbb{P}\left[ A(E,\alpha,s,v) \neq \varnothing\right] \\
      & \leq \sum_{X \in \cC_s({v})} (1-p)^{|\partial X|} \sum_{m \geq \alpha s} \binom{s}{m} p^m (1-p)^{s-m}\\
      & \leq \sum_{X \in \cC_s({v})} (1-p)^{|\partial X|} 2^{-s D(\alpha\|p)} \text{\qquad by \Cref{chernoff}}\\
      & = \sum_{X \in \cC_s({v})} \left(\frac{1-p}{1- \frac{1}{d-1}}\right)^{|\partial X|} \frac{2^{-s D(\alpha\|p)}}{(d-1)^s} (d-1)^s (1-\frac{1}{d-1})^{|\partial X|}\\
      & \leq \left(\frac{1-p}{1-\frac{1}{d-1}}\right)^{(d - 2)s + 2} 2^{-s D(\alpha \| p)} (d-1)^{s} \sum_{X \in \cC_s({v})} \left(\frac{1}{d-1}\right)^s \left(1-\frac{1}{d-1}\right)^{|\partial X|} \text{\qquad since } p < \frac{1}{d-1} \\
      & \leq \left(\frac{(1-p)(d-1)}{d-2}\right)^{2} \left[\left(1-p\right)^{d - 2} 2^{-D(\alpha \| p)} \frac{(d-1)^{d-1}}{(d-2)^{d-2}}\right]^s \text{\qquad from Eq.}  \eqref{eqn:alpha=1}\\
      & = \left(\frac{d-1}{d-2}\right)^{2} q^s
    \end{align*}
    where we defined $q = (1-p)^{d-1-\alpha} p^{\alpha} 2^{h(\alpha) + (d-1) h(\frac{1}{d-1})}$. The main improvement compared to the local stochastic bound is  the factor $(1-p)^{d-1-\alpha}$.

    Summing over $s \geq t$ and $v \in V$ yields:
    \begin{align*}
      \mathbb{P}[\textrm{MaxConn}_{\alpha}(E) \geq t]
      & \leq \sum_{s \geq t} \sum_{v \in V} \left(\frac{d-1}{d-2}\right)^{2} q^s \\
      & \leq |V| \left(\frac{d-1}{d-2}\right)^{2} \frac{q^t}{1-q} \ .
    \end{align*}

  \end{proof}
  
  \begin{lemme}
    \label{lem:threshold-iid}
    For  $\alpha \in (0,1]$, the function $q : p \mapsto (1-p)^{d-1-\alpha} p^{\alpha} 2^{h(\alpha)}(d-1)(1+\frac{1}{d-2})^{d-2}$ is increasing in the interval $[0, \frac{\alpha}{d-1}]$ and it satisfies $q(0) = 0$ and $q(\frac{\alpha}{d-1}) \geq 1$. Consequently, there exists $p_{\mathrm{iid}}$ in this interval such $q(p_{\mathrm{iid}})=1$.
  \end{lemme}
  \begin{proof}
    It is simple to verify that $q$ is increasing. Moreover, we have $q(0) = 0$ and
    \begin{align*}
      q(\frac{\alpha}{d-1}) &= \left(1- \frac{\alpha}{d-1}\right)^{d-1-\alpha} \left(\frac{\alpha}{d-1}\right)^{\alpha} 2^{h(\alpha)} (d-1) \left(1+\frac{1}{d-2}\right)^{d-2} \\
      &= \left(1- \frac{\alpha}{d-1}\right)^{d-1-\alpha} \left(\frac{\alpha}{d-1}\right)^{\alpha} 2^{h(\alpha)}  \frac{(d-1)^{d-1}}{(d-2)^{d-2}}\\
      &= \frac{(d-1-\alpha)^{d-1-\alpha}}{(d-2)^{d-2}} \alpha^{\alpha} \alpha^{-\alpha} (1-\alpha)^{-(1-\alpha)} \\
      &\geq 1 \ ,
    \end{align*}
    where we used the fact that $\alpha \in (0,1]$.
  \end{proof}

  The following proposition gives an upper bound on the value of $p$ for which we have only small connected components. This shows basically that the value in Eq.~\eqref{eq:pls} of Theorem~\ref{thm:percolation} can be at most improved by a factor of $(\frac{2^{-h(\alpha)}}{e})^{\frac{1}{\alpha}}$.

  \begin{proposition}
    \label{prop:ub-percolation}
    Let $d \geq 3$ be an integer, $\alpha \in (0,1]$, $\ell$ be the smallest integer $\ell > \frac{1}{\alpha}$ and $p > \frac{1}{(d-1)^{\ell}}$. Then there exists a constant $\kappa$ (only depending on $d, \alpha, p$) such that for any $c \in \bN^*$, there exists and a family of graphs $G_k$ on $N_k$ vertices (with $N_k \to \infty$) such that if $E$ contains each vertex independently with probability $p$, we have
      \begin{align*}
        \lim_{k \to \infty} \mathbb{P}[\mathrm{MaxConn}_{\alpha}(E) \geq \kappa |N_k|^{1-1/c}] = 1 \ . 
      \end{align*}
  \end{proposition}
  \begin{proof}
    We define $G_k$ to be a complete $(d-1)$-tree of height $c k$. Note that every vertex has degree at most $d$ and the number of vertices is $N_k = \frac{(d-1)^{ck} - 1}{d-2}$. We write $V$ for the set of vertices of the graph $G_k$ and $L_{i}$ for the set of vertices at depth $i$, where the root is at depth $0$.

    For every such vertex $v \in L_{(c-1)k}$, we define the event:
    \begin{align}
      \label{eq:def-ev}
      E_{v} = \left\{ \exists w_{(c-1)k}, \dots, w_{ck-1} : w_{(c-1)k} = v \text{, $w_{i+1}$ is a child of $w_i$ and } |E \cap \{w_{(c-1)k}, \dots, w_{ck-1}\}| \geq \frac{1}{\ell} k \right\}.
    \end{align} 
    We define $P_{v} = \varnothing$ if $E_v$ does not hold and if it does hold, then $P_{v} = \{ w_{(c-1)k}, \dots, w_{ck-1} \}$, where $w_i$ are as in~\eqref{eq:def-ev}.
    We then define the (random) subset $S$ of vertices as
    \begin{align*}
      S = \bigcup_{i=0}^{(c-1)k-1} L_i \cup \bigcup_{v \in L_{(c-1)k}} P_{v} \ .
    \end{align*}
    The claim is then that, provided $k$ is large enough, with high probability, $S$ is an $\alpha$-subset of $E$ and it is of size $\Omega(|N_k|^{1-1/c})$. To see this, observe that
    \begin{align*}
      |E \cap S| &\geq |\{ v \in L_{(c-1)k} : E_{v} \text{ holds}\}| \frac{1}{\ell} k \\
      &= \alpha k |\{ v \in L_{(c-1)k} : E_{v} \text{ holds}\}| + \left(\frac{1}{\ell} - \alpha\right) k |\{ v \in L_{(c-1)k} : E_{v} \text{ holds}\}|
    \end{align*}
    On the other hand
    \begin{align*}
      |S| = k |\{ v \in L_{(c-1)k} : E_{v} \text{ holds}\}| + \frac{(d-1)^{(c-1)k} - 1}{d-2} \ .
    \end{align*}

    Thus, in order to show that $S$ is an $\alpha$-subset of $E$, it suffices to show that for large enough $k$, we have with high probability
    \begin{align*}
      \left(\frac{1}{\ell} - \alpha\right) k |\{ v \in L_{(c-1)k} : E_{v} \text{ holds}\}| \geq \alpha \frac{(d-1)^{(c-1)k} - 1}{d-2} \ .
    \end{align*}
    Now to prove this, it suffices to show that for any $v \in L_{(c-1)k}$, $\mathbb{P}[E_{v}] \geq \tau$ with $\tau > 0$ a constant independent of $k$ and $c$. In fact, once we have that, $|\{ v \in L_{(c-1)k} : E_{v} \text{ holds}\}|$ has a binomial distribution with parameters $((d-1)^{(c-1)k}, \geq \tau)$ and the inequality will follow easily for large enough $k$. Note also the inequality $\mathbb{P}[E_{v}] \geq \tau$ also implies the claimed lower bound on the size of $S$.

    Now to lower bound the probability of $E_{v}$, we look at an even more restricted event. Namely, we consider the event that $w_{(c-1)k} = v \in E$, and there exists a descendent $w_{(c-1)k+\ell} \in L_{(c-1)k+\ell}$ of $w_{(c-1)k}$ that is in $E$, and there is a descendent $w_{(c-1)k+2\ell} \in L_{(c-1)k+2\ell}$ of $w_{(c-1)k+1}$ that is in $E$, and all the way until we reach the bottom of the tree. Note that the probability that such a path exists is lower bounded by the probability of survival of the branching process with offspring distribution binomial with parameters $((d-1)^{\ell}, p)$. But as $p (d-1)^{\ell} > 1$, the probability of survival is $>0$ and we obtain the desired result. 
  \end{proof}

  \section*{Acknowledgments}
  
  We would like to thank Benjamin Audoux, Alain Couvreur, Anirudh Krishna, Vivien Londe, Jean-Pierre Tillich and Gilles Z\'emor for many fruitful discussions on quantum codes as well as Daniel Gottesman for answering our questions about \cite{gottesman2013fault}.
  AG and AL acknowledge support from the ANR through the QuantERA project QCDA.


\begin{thebibliography}{10}

\bibitem{adler1991bootstrap}
Joan Adler.
\newblock Bootstrap percolation.
\newblock {\em Physica A: Statistical Mechanics and its Applications},
  171(3):453--470, 1991.

\bibitem{aharonov1997fault}
Dorit Aharonov and Michael Ben-Or.
\newblock Fault-tolerant quantum computation with constant error.
\newblock In {\em Proceedings of the twenty-ninth annual ACM symposium on
  Theory of computing}, pages 176--188. ACM, 1997.

\bibitem{bravyi2009no}
Sergey Bravyi and Barbara Terhal.
\newblock A no-go theorem for a two-dimensional self-correcting quantum memory
  based on stabilizer codes.
\newblock {\em New Journal of Physics}, 11(4):043029, 2009.

\bibitem{calderbank1996good}
A~Robert Calderbank and Peter~W Shor.
\newblock Good quantum error-correcting codes exist.
\newblock {\em Physical Review A}, 54(2):1098, 1996.

\bibitem{DN17}
Nicolas Delfosse and Naomi~H. Nickerson.
\newblock Almost-linear time decoding algorithm for topological codes.
\newblock {\em arXiv preprint arXiv:1709.06218}, 2017.

\bibitem{delfosse2010quantum}
Nicolas Delfosse and Gilles Z{\'e}mor.
\newblock Quantum erasure-correcting codes and percolation on regular tilings
  of the hyperbolic plane.
\newblock In {\em Information Theory Workshop (ITW), 2010 IEEE}, pages 1--5.
  IEEE, 2010.

\bibitem{delfosse2013upper}
Nicolas Delfosse and Gilles Z{\'e}mor.
\newblock Upper bounds on the rate of low density stabilizer codes for the
  quantum erasure channel.
\newblock {\em Quantum Information \& Computation}, 13(9\&10):0793--0826, 2013.

\bibitem{dennis2002topological}
Eric Dennis, Alexei Kitaev, Andrew Landahl, and John Preskill.
\newblock Topological quantum memory.
\newblock {\em Journal of Mathematical Physics}, 43(9):4452--4505, 2002.

\bibitem{edmonds1965maximum}
Jack Edmonds.
\newblock Maximum matching and a polyhedron with 0, 1-vertices.
\newblock {\em Journal of Research of the National Bureau of Standards B},
  69(125-130):55--56, 1965.

\bibitem{freedman2002z2}
Michael~H Freedman, David~A Meyer, and Feng Luo.
\newblock Z2-systolic freedom and quantum codes.
\newblock {\em Mathematics of quantum computation, Chapman \& Hall/CRC}, pages
  287--320, 2002.

\bibitem{gallager1962low}
Robert Gallager.
\newblock Low-density parity-check codes.
\newblock {\em IRE Transactions on information theory}, 8(1):21--28, 1962.

\bibitem{gottesman1997stabilizer}
Daniel Gottesman.
\newblock Stabilizer codes and quantum error correction.
\newblock {\em arXiv preprint quant-ph/9705052}, 1997.

\bibitem{gottesman2013fault}
Daniel Gottesman.
\newblock Fault-tolerant quantum computation with constant overhead.
\newblock {\em arXiv preprint arXiv:1310.2984}, 2013.

\bibitem{guth2014quantum}
Larry Guth and Alexander Lubotzky.
\newblock Quantum error correcting codes and 4-dimensional arithmetic
  hyperbolic manifolds.
\newblock {\em Journal of Mathematical Physics}, 55(8):082202, 2014.

\bibitem{hastings2013decoding}
Matthew~B Hastings.
\newblock Decoding in hyperbolic spaces: Ldpc codes with linear rate and
  efficient error correction.
\newblock {\em arXiv preprint arXiv:1312.2546}, 2013.

\bibitem{janson2009percolation}
Svante Janson.
\newblock On percolation in random graphs with given vertex degrees.
\newblock {\em Electronic Journal of Probability}, 14:86--118, 2009.

\bibitem{kim2007quantum}
Isaac~Hyun Kim.
\newblock {\em Quantum codes on Hurwitz surfaces}.
\newblock PhD thesis, Massachusetts Institute of Technology, 2007.

\bibitem{kitaev2003fault}
A~Yu Kitaev.
\newblock Fault-tolerant quantum computation by anyons.
\newblock {\em Annals of Physics}, 303(1):2--30, 2003.

\bibitem{kovalev2013fault}
Alexey~A Kovalev and Leonid~P Pryadko.
\newblock Fault tolerance of quantum low-density parity check codes with
  sublinear distance scaling.
\newblock {\em Physical Review A}, 87(2):020304, 2013.

\bibitem{leverrier2015quantum}
Anthony Leverrier, Jean-Pierre Tillich, and Gilles Z{\'e}mor.
\newblock Quantum expander codes.
\newblock In {\em Foundations of Computer Science (FOCS), 2015 IEEE 56th Annual
  Symposium on}, pages 810--824. IEEE, 2015.

\bibitem{londe2017golden}
Vivien Londe and Anthony Leverrier.
\newblock {Golden codes: quantum LDPC codes built from regular tessellations of
  hyperbolic 4-manifolds}.
\newblock {\em arXiv preprint arXiv:1712.08578}, 2017.

\bibitem{lyons1992random}
Russell Lyons.
\newblock Random walks, capacity and percolation on trees.
\newblock {\em The Annals of Probability}, pages 2043--2088, 1992.

\bibitem{pah2015combinatorial}
Chin~Hee Pah and Mohamed~Ridza Wahiddin.
\newblock Combinatorial interpretation of raney numbers and tree enumerations.
\newblock {\em Open Journal of Discrete Mathematics}, 5(01):1, 2015.

\bibitem{richardson2008modern}
Tom Richardson and Ruediger Urbanke.
\newblock {\em Modern coding theory}.
\newblock Cambridge University Press, 2008.

\bibitem{sipser1996expander}
Michael Sipser and Daniel~A Spielman.
\newblock Expander codes.
\newblock {\em IEEE Transactions on Information Theory}, 42(6):1710--1722,
  1996.

\bibitem{steane1996error}
Andrew~M Steane.
\newblock Error correcting codes in quantum theory.
\newblock {\em Physical Review Letters}, 77(5):793, 1996.

\bibitem{tillich2014quantum}
Jean-Pierre Tillich and Gilles Z{\'e}mor.
\newblock Quantum ldpc codes with positive rate and minimum distance
  proportional to the square root of the blocklength.
\newblock {\em IEEE Transactions on Information Theory}, 60(2):1193--1202,
  2014.

\bibitem{uehara1999number}
Ryuhei Uehara et~al.
\newblock The number of connected components in graphs and its applications.
\newblock {\em Manuscript. URL: http://citeseerx. ist. psu.
  edu/viewdoc/summary}, 1999.

\bibitem{wang2009threshold}
David~S Wang, Austin~G Fowler, Ashley~M Stephens, and Lloyd Christopher~L
  Hollenberg.
\newblock Threshold error rates for the toric and surface codes.
\newblock {\em arXiv preprint arXiv:0905.0531}, 2009.

\bibitem{zemor2009cayley}
Gilles Z{\'e}mor.
\newblock {On Cayley Graphs, Surface Codes, and the Limits of Homological
  Coding for Quantum Error Correction}.
\newblock In {\em IWCC}, pages 259--273. Springer, 2009.

\end{thebibliography}
\end{document}